\documentclass[%
 reprint,
 showkeys,
superscriptaddress,
 amsmath,amssymb,
 aps,
 pra,
]{revtex4-1}
\usepackage{amsthm}
\usepackage{graphicx}
\usepackage{dcolumn}
\usepackage{bm}
\usepackage[export]{adjustbox} 
\usepackage{booktabs}
\usepackage{makecell}
\usepackage{hyperref}
\usepackage[braket, qm]{qcircuit}
\usepackage{subcaption, wasysym}
\usepackage{xcolor}

\hypersetup{
	colorlinks,
	linkcolor={red!50!black},
	citecolor={blue!50!black},
	urlcolor={blue!80!black}
}

\newcommand{\circspace}{\vspace{0.5cm}}

\newtheorem{theorem}{Theorem}[section]

\theoremstyle{definition}
\newtheorem{definition}{Definition}[section]

\begin{document}


\title{New techniques for fault-tolerant decomposition of Multi-Controlled Toffoli gate}

\author{Laxmidhar Biswal}
 \affiliation{School of VLSI Technology, IIEST Shibpur, India }
\author{Debjyoti Bhattacharjee}%
 \email{debjyoti001@ntu.edu.sg}
\affiliation{%
 School of Computer Science and Engineering, NTU, Singapore
}%
\author{Anupam Chattopadhyay}
\affiliation{%
 School of Computer Science and Engineering, NTU, Singapore
}
\author{Hafizur Rahaman}
\affiliation{School of VLSI Technology, IIEST Shibpur, India }

\date{\today}

\begin{abstract}
Physical implementation of scalable quantum architectures faces an immense challenge in form of fragile quantum states. To overcome it, quantum architectures with fault tolerance is desirable. This is achieved currently by using surface code along with a transversal gate set. This dictates the need for decomposition of universal Multi Control Toffoli~(MCT)
gates using a transversal gate set. Additionally, the transversal non-Clifford phase gate incurs high latency which makes it an important factor to consider during decomposition.   
Besides, the decomposition of large Multi-control Toffoli~(MCT) gate without ancilla presents an additional hurdle. In this manuscript, we address both of these issues by introducing Clifford+$Z_N$ gate library. We present an ancilla free decomposition of MCT gates with linear phase depth and quadratic phase count. Furthermore, we provide a technique for decomposition of MCT gates in unit phase depth using the Clifford+$Z_N$ library, albeit at the cost of ancillary lines and quadratic phase count. 


\end{abstract}

\pacs{Valid PACS appear here}
\keywords{\em Clifford+$Z_N$, Quantum Computing, Multi-control Toffoli gate, Decomposition}
\maketitle

\section{Introduction}
\noindent 
Quantum computing brings forth the promise of quantum supremacy, which states that for certain classes of problems the efficiency of quantum computer cannot ever be matched by a classical computer~\cite{Shor,Grover}.
However, practical realization of quantum computer is  immensely  challenging as  quantum  circuit does  not  allow  fan-out  and  direct  feedback  due  to  limitation posed by the  no-cloning  theorem~\cite{Buzek}. Besides, the physical hardware realizing quantum computing is highly susceptible towards external noise, which corrupts the information being processed~\cite{Ladd}.  

Quantum   states  are fragile  with  low  coherence  time that makes it difficult to keep the coherent superposition of quantum information state~(qubit) intact~\cite{Steane}. This drives the need for fault tolerant circuits. Fault tolerant quantum circuit can be realized by encoding quantum states using quantum error correction code~(QECC) of which surface code is the most promising one. The key objective of QECC is to prevent   propagation  of  errors  within  the  encoded blocks. This permits each small error  to be dealt independently with block-wise separation  in order to  hold  the  gate  error-rate  within the  threshold  limits~\cite{Paetznic}. To accomplish this, the universal set of Clifford+T-group of gates consisting of H-gate, CNOT-gate, phase~(S-gate) and non-Clifford phase gate~(T-gate), is used. 
Surface code along with Clifford+T-group, has become the de-facto technology in the field of design and implementation of fault-tolerant quantum circuits~\cite{Cody_Jones_2012}.

The $n$-qubit Multi-Controlled Toffoli~(MCT) gate plays pivotal role in the implementation of quantum algorithms~\cite{Miller}.  Barrenco~\textit{et al.} presented an ancilla-free technique for the decomposition of the MCT-gate into set of roots of Pauli's~\textit{X} in exponential depth using~\cite[Lemma~7.1]{Barrenco}, {with help of ancilla inputs} in polynomial depth~\cite[Lemma-~7.3]{Barrenco} and also in quadratic depth~\cite[Lemma~7.6]{Barrenco}.  Saeedi~\textit{et al.} has presented a linear depth
ancilla-free decomposition of MCT gate using rotation-based gate~(\textbf{$R_\theta$})~\cite{Saeedi}. In fact, the decomposition approaches 	presented in the literature so far are unable to protect threshold limit of QECC due to lack of transversal gates in the design approach. However, transversal gate is necessary to achieve large scale fault-tolerant~\cite{Eastin} as, the threshold limit of most promising surface code is only 0.57\% per physical qubit and 0.75\% per operation~\cite{Fowler-1}.
\begin{table}[t]
\circspace
\centering
\caption{\em Elementary quantum gates and their operations.}
\setlength\extrarowheight{4pt}
\begin{tabular}{ rccp{3.5cm}  } \hline
	
 \multicolumn{1}{c}{\textbf{Gate}} & 
 \textbf{Symbol}  & 
 \multicolumn{1}{c}{\textbf{Matrix}}
   & \multicolumn{1}{c}{\textbf{Operations} }\\ \hline
 NOT & \Qcircuit @C=.7em @R=.7em {
& \targ & \qw
} & \begin{adjustbox}{valign=t}$\begin{bmatrix}
    0       &  1 \\
    1       &  0 
     \end{bmatrix}$ \end{adjustbox} & X$\ket{0}=\ket{1}$ \newline X$\ket{1}=\ket{0}$ \\
 CNOT & \Qcircuit @C=.7em @R=.5em {
&  \ctrl{1} & \qw\\ & \targ &  \qw \\} & 
   \begin{adjustbox}{valign=t}
    $\begin{bmatrix}
    1 & 0 & 0 & 0 \\
    0 & 1 & 0 & 0 \\
    0 & 0 & 0 & 1 \\
    0 & 0 & 1 & 0 \\
     \end{bmatrix}$ \end{adjustbox}
     &

    $\ket{00} \rightarrow \ket{00}$ \newline 
    $\ket{01} \rightarrow \ket{01}$  \newline
    $\ket{10} \rightarrow \ket{11}$ \newline
    $\ket{11} \rightarrow \ket{10}$  \\
  $Z$ gate & \Qcircuit @C=.7em @R=.7em {
& \gate{Z} & \qw
} &
\begin{adjustbox}{valign=t}$
\begin{bmatrix}
    0       & 0 \\
    0       & -1 \\
     \end{bmatrix}$ 
     \end{adjustbox} & X$\ket{0}=\ket{0}$ \newline X$\ket{1}=-\ket{1}$ \\
S gate & \Qcircuit @C=.7em @R=.5em {& \gate{S} & \qw} & $\begin{bmatrix} 1 & 0 \\ 0 & i \\
\end{bmatrix}$ & S$\ket{0} = \ket{0}$ \newline S$\ket{1} = e^\frac{i \pi}{2}\ket{1}$ \\

     $ S^\dag $gate & \Qcircuit @C=.7em @R=.5em {& \gate{S^\dag} & \qw} 
     & \begin{adjustbox}{valign=t} 
     $\begin{bmatrix} 
     1 & 0 \\ 
     0 & e^\frac{-i \pi}{2} 
\end{bmatrix}$\end{adjustbox} &
$S^\dag \ket{0} =\ket{0}$ \newline $ S^\dag \ket{1} = e^\frac{-i \pi}{2}\ket{1}$ \\
  
$T$ gate & \Qcircuit @C=.7em @R=.5em {& \gate{T} & \qw} & 
\begin{adjustbox}{valign=t}
$\begin{bmatrix} 1 & 0 \\ 0 & e^\frac{i \pi}{4} \\
\end{bmatrix}$ \end{adjustbox}
& $T\ket{0} =\ket{0}$  \newline $T\ket{1} = e^\frac{i \pi}{4}\ket{1}$ \\
 $ T^\dag $gate & \Qcircuit @C=.7em @R=.5em {& \gate{T^\dag} & \qw} & 
 \begin{adjustbox}{valign=t}
 $\begin{bmatrix} 1 & 0 \\ 0 & e^\frac{-i \pi}{4} \\
\end{bmatrix}$ \end{adjustbox}& $T^\dag \ket{0} =\ket{0}$  \newline $ T^\dag \ket{1} = e^\frac{-i \pi}{4}\ket{1}$ \\
 
$Z_N$ gate & \Qcircuit @C=.7em @R=.5em {& \gate{Z_N} & \qw} & 
\begin{adjustbox}{valign=t}
$\begin{bmatrix} 1 & 0 \\ 0 & e^\frac{i \pi}{N} \\
\end{bmatrix}$ \end{adjustbox} & $Z_N \ket{0} =\ket{1}$  \newline $ Z_N \ket{1} = e^\frac{i \pi}{N}\ket{1}$ \\
 $Z_N^\dag$ gate & \Qcircuit @C=.7em @R=.5em {& \gate{Z_N^\dag} & \qw} &
 \begin{adjustbox}{valign=t}
 $\begin{bmatrix} 1 & 0 \\ 0 & e^\frac{-i \pi}{N} \\
\end{bmatrix}$ \end{adjustbox} & $Z_N^\dag \ket{0} =\ket{0}$  \newline $ Z_N^\dag \ket{1} = e^\frac{-i \pi}{N}\ket{1}$ \\

 Hadamard(H) & \Qcircuit @C=.7em @R=.5em {
& \gate{H} & \qw
} & 
\begin{adjustbox}{valign=t}
$ \frac{1}{\sqrt{2}} \begin{bmatrix} 1 & 1 \\ 1 & -1 \\
\end{bmatrix}$ \end{adjustbox}& H $\ket{0} = \frac{1}{\sqrt 2}(\ket{0}+ \ket{1})$  \newline H $\ket{1} = \frac{1}{\sqrt 2}(\ket{0}- \ket{1})$ \\ \hline 
\end{tabular}
\label{table-1}
\end{table}
Decomposition of Toffoli gate(s) with more than two control bit into the  Clifford+$T$-group without ancilla is challenging. This is one of the key challenges that we address in this paper. To address this problem, the existing universal transversal gate set, namely the Clifford+$T$, should be re-defined that would allow decomposition of $n$-MCT gate without any ancillary line.

Low qubit coherence time, {high latency of phase gates} and high gate error rate limits the depth of the circuit especially in phase-depth of the quantum circuit~\cite{Martinis}. Shallow quantum circuits plays an important role in the implementation of larger quantum algorithms~\cite{Bravyi}. This makes the decomposition of $n$-MCT into a constant phase-depth based quantum circuit crucial. This is the second challenge we address in this paper. 

The key contributions of this paper are as follows. 
\begin{itemize}
	\item {A new universal quantum gate library~(Clifford+$Z_N$) is constituted by inclusion of $N^{th}$ root of Pauli's \textit{Z}-gate into existing popular Clifford+$T$-group.}
	\item A linear phase-depth ancilla-free decomposition technique of $n$-MCT gate has been developed using the proposed Clifford+$Z_N$ universal quantum gate library.
	\item A constant unit phase depth decomposition of $n$-MCT using additional ancillary input lines has been developed.
\end{itemize}

The rest of the paper is organized as follows. In section~\ref{sec:prelim}, the preliminaries related to quantum gates are presented. Section~\ref{sec:method} presents our proposed solution to decompose $n$-MCT with linear phase count. 
Section~\ref{sec:unitph} presents the decomposition method to achieve unit phase depth realization of $n$-MCT gate, followed by a brief summary of existing results in Section~\ref{sec:exp}. Section~\ref{sec:conc} briefly concludes the paper.

\section{Preliminaries}\label{sec:prelim}
\noindent In this section, we present the preliminaries of quantum gates and their associated properties. Table.~\ref{table-1} presents some important primitive quantum gates.

\begin{definition}[n-MCT gate]
	A $n$-MCT gate consists of $n-1$ control lines~[$c_1, c_2, \ldots ,c_{n-1}]$ and a single target line~$t$. The state of the target line changes based on the value of the control inputs.
	\begin{equation}\label{eq:targ-MCT}
	t = (c_1.c_2\ldots c_{n-1}) \oplus c_n
	\end{equation}
\end{definition}

\begin{definition}[n-MCZ gate]
	A $n$-MCZ gate is a conditional sign gate that changes the sign of the target line~$t$ with respect to the state of the all control lines and shifts the phase of the quantum states by an angle of $\pi$-degrees in Z-direction when all the $n-1$ control bits are $1$, as described in equation~\eqref{eq:nmcz}.\\
	\begin{equation} \label{eq:nmcz}
	t=(-1)^{c_1.c_2 \ldots c_{n-1}.c_{n}}
	\end{equation}
\end{definition}
In general, we denote a gate~$G$ with $k$-control lines as $C^k(G)$. For example, a two control $Z$ gate would be denoted as $C^2(Z)$. 

\begin{definition}[Pauli-gate] The Pauli-gates are a set of three 2x2 complex matrices which are Hermitian and unitary and takes into account the interaction of the spin of a particle with an external electromagnetic field. \\
		$X$, $Y$, and $Z$ represents Pauli's operator in $x$-axis, $y$-axis and $z$-axis. The $I$ represents Identity matrices.\\ 
	\begin{equation}
	X.X=Y.Y=Z.Z=I ; XYZi=I
	\end{equation}
\end{definition}

\begin{definition}[Transversal gate] Transversal gates operate bit-wise within encoded block of QECC and can interact with the corresponding qubit, either in another block or in a specialized ancilla. \end{definition}

\begin{definition}[Clifford group] The Clifford group is a set of special kind of quantum gates (G) which satisfies 
\begin{equation}
G^{\dag}PG \in P^{(\oplus)n}
\end{equation}
where $P$ represents Pauli-gate and  $P \in \{I, X, Y, Z\}$. \dag represents the self inverse of a gate.  The Clifford group is composed with \textit{H} gate, Pauli's matrices $\{X, Y, Z \}$ and \textit{S} gate along with \textit{CNOT} gate.
\end{definition}
Each logical gate in Clifford group~(NOT, CNOT, H, Z, X, Y and S) is transversal. Transversal operators do not propagate errors between the lines within the same encoded block of QECC. Any quantum circuit built over only transversal gates are inherently fault tolerant\cite{Gottesman}.

\begin{definition}[] 
	$Z_N$ defines a non-Clifford quantum operator, that represents $N_{th}$ root of Pauli-\textit{Z} gate, i.e., $Z_N = \sqrt[N]{Z}$, where $N=2^n$ and $n \in \{1,2,\ldots,\}$.
\end{definition}

For fault tolerant quantum computation, a pure magical $Z_N$-state has to be generated. The concept of magic state distillation is related to the subject of noise. Magic state distillation yields a pure quantum state with the required fidelity. It is stated that each distillation level requires $(4N-1)$ number of input states to produce a single pure magical output state~\cite{Mishra, Landahl}. 
Roots of the interrelated Pauli matrices have inherent unitary properties as well as ability to make the circuit fault tolerant one. This makes the roots of
interrelated Pauli matrices in the design of quantum computer. The inter relation between the root of the Pauli-\textit{X} and Pauli-\textit{Z} is given by following equation.
\begin{equation}\label{eq:root}
\sqrt[K]{X}= H.{\sqrt[K]{Z}}.H
\end{equation} 

The functionality of both $T$-gate and $Z_N$ gate are analogous. These gates are used to flip the phase of the qubit in $Z$-direction. On application of $T$-gate, the phase of the qubit flips by an angle of $\frac{\pi}{4}$; whereas $Z_N$-gate flips the phase by an angle $\frac{\pi}{N}$. In particular for N=4, $Z_N=T$. Therefore, we  estimate the performance of the quantum circuit in terms of both phase-count and phase-depth~(comparable to T-count and  T-depth), as non-Clifford phase gate ($Z_N$) has high distillation cost and high execution time~\cite{Fowler-Stephens, Amy-Maslov}. 
\begin{definition}[Phase-count] The minimum number of unitary phase gates ($Z_N$)used in the design of quantum circuit is defined as the phase count.
	\end{definition}
\begin{definition}[Phase-depth] Phase depth is the minimum number of cycles required to execute all the $Z_N$ gates. 
	\end{definition}

\section{Proposed Methodology}\label{sec:method}
\noindent This section provides a new technique for decomposition of $n$-MCT using the Clifford+$Z_N$-group, introduced in the next subsection.
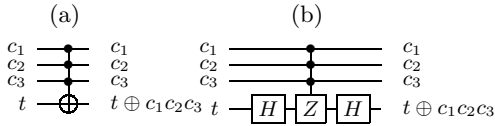
\begin{figure}[t]
	\circspace
	\begin{subfigure}[t]{0.3\columnwidth}
		\caption{}
		\label{subfig:4MCT}

			\scalebox{0.9}{
				\Qcircuit  @C= 0.5em @R=.5em {
					\lstick{c_1} & \qw & \ctrl{3} & \qw & \rstick{c_1} \\
					\lstick{c_2} & \qw  & \ctrl{0} & \qw & \rstick{c_2} \\
					\lstick{c_3} & \qw & \ctrl{0} & \qw & \rstick{c_3}\\
					\lstick{t} & \qw & \targ & \qw & \rstick{t\oplus {{c_1}{c_2}{c_3}}}}
			}

	\end{subfigure}
	\begin{subfigure}[t]{0.42\columnwidth}
		\caption{ }
		\label{Subfig:4MCT-MCZ}
		\scalebox{0.9}{
			\Qcircuit  @C= 0.5em @R=.5em {
				\lstick{c_1} & \qw  & \qw & \ctrl{3}  & \qw & \qw & \rstick{c_1} \\
				\lstick{c_2} & \qw  & \qw & \ctrl{0}  & \qw & \qw & \rstick{c_2} \\
				\lstick{c_3} & \qw  & \qw & \ctrl{0}  & \qw & \qw & \rstick{c_3} \\
				\lstick{t} & \qw & \gate{H} & \gate{Z} & \gate{H} & \qw & \rstick{t\oplus {{c_1}{c_2}{c_3}}}}
		}
	\end{subfigure}
	\caption{\em (\subref{subfig:4MCT})~4-MCT gate.
		(\subref{Subfig:4MCT-MCZ})~Decomposition of $4$-MCT using $4$-MCZ gate and $H$-gates }
	\label{fig:4mct}
\end{figure}

\begin{figure}[t]
	\circspace
	\centerline{
		\Qcircuit  @C= 0.3em @R=.5em {
			& & & & & & & & & & \lstick{c_1} &\qw & \qw & \ctrl{3} & \qw & \ctrl{2} & \qw & \ctrl{2} & \qw & \qw & \rstick{c_1} \\
			& & & & & & & & & & \lstick{c_2} & \qw & \qw & \ctrl{0} & \qw & \ctrl{0} & \qw & \ctrl{0} & \qw & \qw & \rstick{c_2} \\
			& & & & & & & & & & \lstick{c_3} &  \qw & \qw & \qw & \ctrl{1} & \targ & \ctrl{1} & \targ & \qw & \qw & \rstick{c_3} \\
			& & & & & & & & & & \lstick{t} & \qw & \gate{H} & \gate{Z_2} & \gate{Z_2} & \qw & \gate{Z^\dag_2}& \gate{H} & \qw & \qw & \qw & \qw & \rstick{t\oplus {{c_1}{c_2}{c_3}}} \\
	}}
	\caption{\em Circuit obtained by the decomposition of $C^3(Z)$ gate in Fig.~\ref{fig:4mct} into $C^2(Z_2)$ and Toffoli~gates.} 	
	\label{fig:decomposed}
	\circspace
	\centerline{
		\scalebox{0.8}{
			\Qcircuit  @C= 0.05em @R=.4em {
				& \qw & \qw &\qw & \qw & \ctrl{3} & \qw & \qw & \qw & \ctrl{2} & \qw & \ctrl{1} & \qw & \qw & \qw & \qw & \qw & \qw & \ctrl{1} & \qw & \ctrl{2} & \qw & \qw & \qw\\
				& \qw & \qw &\qw & \qw & \qw & \ctrl{2} & \qw & \qw & \qw & \ctrl{1} & \targ & \ctrl{1} & \qw & \ctrl{2} & \qw & \qw & \ctrl{1} & \targ & \ctrl{1} & \qw & \qw & \qw & \qw\\
				& \qw & \qw &\qw & \qw & \qw & \qw & \ctrl{1} & \gate{H} & \gate{Z^\dag_2} & \gate{Z^\dag_2} & \qw & \gate{Z_2}& \gate{H} & \qw & \ctrl{1} & \gate{H} & \gate{Z^\dag_2} & \qw & \gate{Z_2} & \gate{Z_2} & \gate{H} & \qw & \qw\\
				& \qw & \qw &\qw & \gate{H} & \gate{Z_4} & \gate{Z_4} & \gate{Z_2} & \qw & \qw & \qw & \qw & \qw & \qw & \gate{Z^\dag_4} & \gate{Z^\dag_2}& \gate{H} & \qw & \qw & \qw & \qw & \qw & \qw & \qw \\
	}}}
	\caption{\em 
		Circuit obtained by the decomposition of $C^2(Z)$ and Toffoli gates in Fig.~\ref{fig:decomposed} into $C^1(Z_4)$ and CNOT~gates.}
	\label{fig:eqcirc}
\end{figure}

\begin{figure*}[t]
	\circspace
	\begin{center}
		\centerline{
			\Qcircuit  @C= 0.5em @R=.1em {
				& & & & \qw & \gate{Z_8} & \targ & \gate{Z^\dag_8} & \targ & \qw & \targ & \gate{Z_4} & \ctrl{1} & \targ & \qw & \qw & \qw & \qw & \qw & \qw & \qw & \qw &\qw & \qw & \targ & \gate{Z^\dag_4} & \ctrl{1} & \targ & \qw & \qw & \qw & \qw 
				& \qw \\
				& & & & \qw & \gate{Z^\dag_8} & \targ & \gate{Z^\dag_8} & \targ & \qw & \targ & \gate{Z_4} & \targ & \qw & \ctrl{1} & \gate{Z^\dag_8} & \ctrl{1} & \qw & \targ & \qw & \gate{Z_8} & \qw & \targ & \qw & \targ & \gate{Z_4} & \targ & \qw & \ctrl{1} & \gate{Z_4} & \ctrl{1} & \qw & \qw\\
				& & & & \qw & \gate{Z_4} & \targ & \gate{Z^\dag_4} & \targ & \gate{H} & \ctrl{-2} & \gate{Z^\dag_4} & \qw & \ctrl{-2} & \targ & \gate{Z^\dag_4} & \targ & \gate{H} & \qw & \ctrl{1} & \gate{Z^\dag_4} & \ctrl{1} & \qw & \gate{H} & \ctrl{-2} & \gate{Z_4} & \qw & \ctrl{-2} & \targ & \gate{Z^\dag_4} & \targ & \gate{H} & \qw\\
				& & & & \qw & \gate{H} & \ctrl{-3} & \gate{Z_8} & \ctrl{-3} & \qw & \qw & \qw & \qw & \qw & \qw & \qw & \qw & \qw & \ctrl{-2} & \targ & \gate{Z_4} & \targ & \ctrl{-2} & \gate{H} & \qw  & \qw  & \qw  & \qw  & \qw  & \qw  & \qw  & \qw & \qw}}
	\end{center}
	\caption{\em Fault tolerant circuit for $4$-MCT gate, using the proposed Clifford+$Z_N$ group.}
	\label{fig:prop4mct}
\end{figure*}
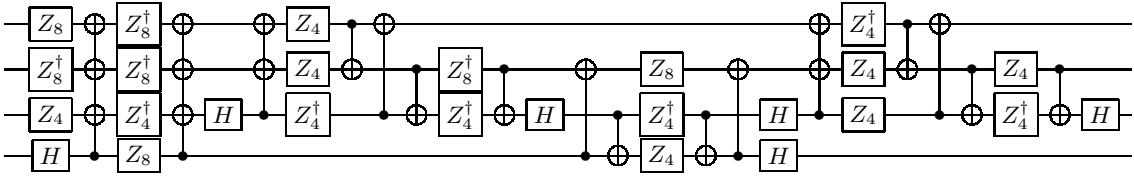

\subsection{Delineation of Clifford+Z\textsubscript{N}-group} \label{Del:ZN}
We introduce a super set of the fault-tolerant Clifford+$T$-group, defined as the Clifford+$Z_N$ group.

\begin{definition}[Clifford+$Z_N$ group]
The Clifford+$Z_n$ group is a fault tolerant quantum gate library consisting of Hadamard gate~H, phase gate~S, controlled not~CNOT, T-gate and $Z_N$-gate, where  $Z_N$ represents $N^{th}$ root of Pauli's $Z$-gate.
\end{definition}
The Clifford+$Z_N$ group is identical to the Clifford+T group for $N=4$. We envision the Clifford+$Z_N$ as the defacto fault tolerant quantum gate library, since this library offers the possibility of ancilla-free fault-tolerant decomposition of $n$-MCT. In the next subsection, we present the decomposition of a 4-MCT gate using the Clifford+$Z_n$ group.


\subsection{Decomposition of a $4$-MCT using Clifford+$Z_n$ group}
Before presenting the general method to decompose \mbox{$n$-MCT} gate, we demonstrate the decomposition of $4$-MCT gate, presented in Fig.~\ref{fig:4mct}, into fault-tolerant structure using  the newly proposed Clifford+$Z_N$ group.  The entire procedure comprises of four steps.
\begin{itemize}
    \item[Step 1.] The $4$-MCT is decomposed into a pair of $H$ gates and $4$-MCZ gate, as shown in Fig.~\ref{Subfig:4MCT-MCZ}~\cite{Barrenco}. The result of this decomposition results in a multi-control $Z$-gate, which can be further decomposed into  root of $Z$-gates with lower number of controls. 
    \item[Step 2.] The $4$-MCZ gate obtained at the end of step~1 has three control lines.
    The three control lines are split into two different set of control line --- the control line~($c_3$) represents one set and remaining two control lines~($c_1,c_2$) belong to a separate set. 
    Now, $4$-MCZ gate can be decomposed into the second root of $Z$-gate; where $Z = S^2 = Z^2_2$. The resulting decomposition structure is presented in Fig.~\ref{fig:decomposed}.
    \item[Step 3.] Fig.~\ref{fig:decomposed} is composed of  $C^2(Z_2)$ gates, Toffoli gates, alongside  one qubit and two qubit quantum gates. The $C^2(Z_2)$  and Toffoli gate are decomposed further using  step~1 and 2, recursively till the circuit comprises of 1 or 2 qubit quantum phase gates, H-gate and CNOT-gates only. The resulting circuit after decomposition of Fig.~\ref{fig:decomposed} is shown in Fig.~\ref{fig:eqcirc}.
    \item[Step 4.]The decomposed circuit in Fig.~\ref{fig:decomposed} consists of single control phase gates. To achieve fault tolerance, the single control phase gates must be decomposed into control-free phase gates. Each 2-qubit phase gate is therefore decomposed using transversal CNOT-gate and control-free phase gate using the technique shown in Fig.~\ref{fig:thm1}. 
\end{itemize}
Subsequently, two optimization rules are applied to the circuit in order to achieve low phase-count and phase-depth. Two adjacent gates which are self inverse to each other, are removed as their net effect is identity. The phase gates are free to be switched over control line whereas there is conditional restriction for switching over target line. On a target line, the phase-gate can not be switched across any single CNOT-gate but, can switch across two successive CNOT gates with common control input.
This approach is used to increase the parallelism of phase gates are carried  to lower phase depth. The final fault tolerant decomposition of $4$-MCT is presented in Fig~\ref{fig:prop4mct}.

The decomposed fault-tolerant circuit corresponding to $4$-MCT has a phase count of 20 with phase depth of 7. In the next section, we present the fault-tolerant decomposition of $n$-MCT using the Clifford+$Z_N$ group, for any $n \ge 4$.

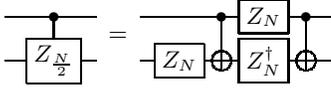
\begin{figure}[t]
	\circspace
	\centerline{
	\Qcircuit @C=0.3em @R=0.1em {
	& \qw & \qw & \ctrl{2} & \qw & \qw & & & & & & & \qw & \qw & \ctrl{2} & \gate{Z_N} & \ctrl{2} & \qw & \qw  \\
	& &  &  &  &  & & & = & & & & & & &  &  &  & \\
	& \qw & \qw & \gate{Z_{\frac{N}{2}}} & \qw & \qw & & & & & & &\qw & \gate{Z_N} & \targ & \gate{Z^\dag_N} & \targ & \qw & \qw
	}
	}
	\caption{\em Decomposition of controlled phase gate into equivalent fault tolerant circuit~\cite{Soeken}.}
	\label{fig:thm1}
\end{figure}


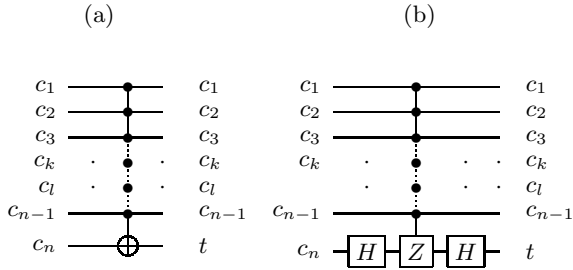
\begin{figure}[t]
	\circspace
	\begin{subfigure}[t]{0.44\columnwidth}
		\centering
		\caption{}
		\label{subfig:n-MCT}
		
		\circspace
		\centerline{
			\Qcircuit @C=1.0em @R=.8em {
				& & \lstick{c_1} & \qw & \ctrl{2} & \qw & \rstick{c_1} \\
				& & \lstick{c_2} & \qw & \ctrl{0} & \qw & \rstick{c_2} \\
				& & \lstick{c_3} & \qw & \ctrl{0} \ar@{.}[ddd] & \qw & \rstick{c_3} \\
				& & \lstick{c_k} & \cdot & \control & \cdot & \rstick{c_k} \\
				& & \lstick{c_l} & \cdot & \control & \cdot & \rstick{c_l} \\
				& & \lstick{c_{n-1}} & \qw & \ctrl{1} & \qw & \rstick{c_{n-1}} \\
				& & \lstick{c_n} & \qw & \targ & \qw & \rstick{t} 
		}}
	\end{subfigure}
	\begin{subfigure}[t]{0.52\columnwidth}
		\centering
		\caption{}
		\label{subfig:n-MCZ}
		\circspace	\centerline{
			\Qcircuit @C=0.6em @R=.8em {
				\lstick{c_1} & \qw & \ctrl{2} & \qw & \qw & \rstick{c_1}\\
				\lstick{c_2} & \qw & \ctrl{0} & \qw & \qw & \rstick{c_2} \\
				\lstick{c_3} & \qw & \ctrl{0} \ar@{.}[ddd]  & \qw & \qw & \rstick{c_3}\\
				\lstick{c_k} & \cdot & \control & \cdot & \cdot & \rstick{c_k} \\
				& \cdot & \control & \cdot & \cdot & \rstick{c_l} \\
				\lstick{c_{n-1}} & \qw & \ctrl{1} & \qw & \qw & \rstick{c_{n-1}}\\
				\lstick{c_n} & \gate{H} & \gate{Z} & \gate{H} & \qw & \rstick{t}
		}}
	\end{subfigure}
	\caption{\em (\subref{subfig:n-MCT}) $n$-MCT gate. 
		(\subref{subfig:n-MCZ}) Decomposition of $n$-MCT gate using $n$-MCZ gate and two H gates.}
	\label{fig:nmct}
\end{figure}
\subsection{Decomposition of $n$-MCT using Clifford+$Z_N$ group}\label{Dec:MCT}
For any arbitrary $n$~($n\ge 4$), we present a method for decomposition $n$-MCT gate using Clifford+$Z_N$-group. 
In the first step, the $n$-MCT gate is translated into the form of Pauli's $Z$-gate through Jacobi method which needs a pair of $H$-gates, as shown in Fig.~\ref{subfig:n-MCZ}~\cite{Barrenco}. 
The $n$-MCZ gate is decomposed into a set of single control phase gates and $k$-MCT gates, where $k <n$, as shown in Fig.~\ref{fig:lowermct}. These two steps are recursively applied to the
present $k$-MCT gates~($k>2$) to decompose them into a set of  single control phase gates and CNOT gates. The result of the decomposition is depicted in Fig.~\ref{fig:intermidiate}. In the final step, each single control phase gate is replaced by the equivalent fault tolerant circuit, using the method shown in Fig.~\ref{fig:thm1}. The final decomposed circuit is shown in Fig.~\ref{fig:decompnMCT}.

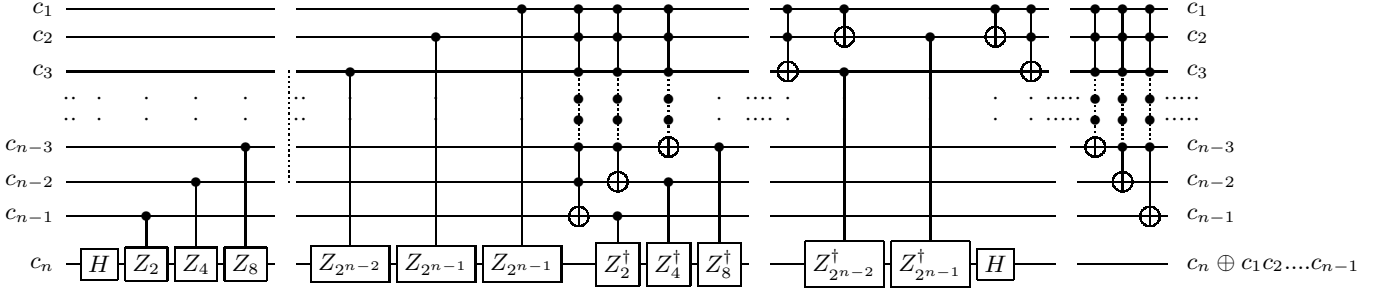
\begin{figure*}[ht]
   \circspace
   \centerline{
   \Qcircuit @C=0.3em @R=0.6em {
   \lstick{c_1} & \qw & \qw & \qw & \qw & \qw & \qw & & & &\qw & \qw & \qw & \ctrl{8} & \ctrl{2} & \ctrl{2} & \ctrl{2} & \qw & \qw & & & & \ctrl{2} & \ctrl{1} & \qw & \ctrl{1} & \ctrl{2} & \qw & & & & \qw & \ctrl{2} & \ctrl{2} & \ctrl{2} & \qw & \rstick{c_1}\\
   \lstick{c_2} & \qw & \qw & \qw & \qw & \qw & \qw & & & &\qw & \qw & \ctrl{7} & \qw & \ctrl{0}  & \ctrl{0} & \ctrl{0}  & \qw & \qw & & & & \ctrl{0} & \targ & \ctrl{7}  & \targ & \ctrl{0} & \qw & & & & \qw & \ctrl{0} & \ctrl{0} & \ctrl{0} & \qw & \rstick{c_2}\\
   \lstick{c_3} & \qw & \qw & \qw & \qw & \qw & \qw & & \ar@{.}[dddd] & & \qw & \ctrl{6} & \qw & \qw & \ctrl{0} \ar@{.}[ddd]& \ctrl{0}\ar@{.}[ddd] & \ctrl{0}\ar@{.}[ddd] & \qw & \qw & & & & \targ & \ctrl{6} & \qw & \qw & \targ & \qw & & & & \qw & \ctrl{0} \ar@{.}[ddd] & \ctrl{0} \ar@{.}[ddd] & \ctrl{0} \ar@{.}[ddd] &\qw &\rstick{c_3}\\
   \cdot & \cdot & \cdot & \cdot & \cdot & \cdot & & & &\cdot & \cdot & \cdot & \cdot & \cdot & \control & \control & \control & \cdot & \cdot & \cdot & \cdot & \cdot & \cdot & \cdot & \cdot & \cdot & \cdot & \cdot & \cdot & \cdot & \cdot & \cdot & \control & \control & \control & \cdot & \cdot & \cdot & \cdot & \cdot \\
   \cdot & \cdot & \cdot & \cdot & \cdot & \cdot & & & & \cdot & \cdot & \cdot & \cdot & \cdot & \control & \control & \control & \cdot & \cdot & \cdot & \cdot  & \cdot & \cdot & \cdot & \cdot & \cdot & \cdot & \cdot & \cdot & \cdot & \cdot & \cdot & \control & \control & \control & \cdot & \cdot & \cdot & \cdot & \cdot\\
   \lstick{c_{n-3}} & \qw & \qw & \qw & \qw & \ctrl{3} & \qw & & & & \qw & \qw & \qw & \qw & \ctrl{2} & \ctrl{1} & \targ & \ctrl{3} & \qw & & & & \qw & \qw & \qw & \qw & \qw & \qw & \qw & & & \qw & \targ & \ctrl{1} & \ctrl{2} & \qw & \rstick{c_{n-3}}\\
   \lstick{c_{n-2}} & \qw & \qw & \qw & \ctrl{2} & \qw & \qw & & & & \qw & \qw & \qw & \qw & \ctrl{0} & \targ & \ctrl{2} & \qw & \qw & & & & \qw & \qw & \qw & \qw & \qw & \qw & \qw & & & & \qw & \targ & \qw & \qw & \rstick{c_{n-2}}\\
   \lstick{c_{n-1}} & \qw & \qw & \ctrl{1} & \qw & \qw & \qw & & & & \qw & \qw & \qw & \qw & \targ & \ctrl{1} & \qw & \qw & \qw & & & & \qw & \qw & \qw & \qw & \qw & \qw & \qw & & & & \qw & \qw & \targ & \qw & \rstick{c_{n-1}}\\
   \lstick{c_n} & \qw & \gate{H} &\gate{Z_2} & \gate{Z_4} & \gate{Z_8} & \qw & & & & \qw & \gate{Z_{2^{n-2}}} & \gate{Z_{2^{n-1}}} & \gate{Z_{2^{n-1}}} & \qw & \gate{Z^\dag_2} & \gate{Z^\dag_4} & \gate{Z^\dag_8} & \qw & & & & \qw & \gate{Z^\dag_{2^{n-2}}} & \gate{Z^\dag_{2^{n-1}}} & \gate{H} & \qw & \qw & \qw & & & & \qw & \qw & \qw & \qw & \rstick{c_n\oplus {c_1}{c_2}....{c_{n-1}}}}
   }
   \caption{\em Decomposition of $n$-MCT circuit in Fig.~\ref{subfig:n-MCZ}
   	into single control phase gates, $k$-MCT gates~($k<n$) and Hadamard gates.}
   \label{fig:lowermct}
   \end{figure*}

	  \begin{figure*}[t]
	\circspace
	\centerline{
		\Qcircuit @C=0.1em @R=0.6em {
			\lstick{c_1} & \qw & \qw & \qw & \qw & \qw & \qw & & & &\qw & \ctrl{8} & \qw & & & \qw & \qw & \qw & & & \qw & \ctrl{2} & \qw & \ctrl{1} & \qw & \qw & \qw & & & \qw & \qw & \qw & \qw & & & \qw & \qw & \ctrl{1} & \qw & \ctrl{2} &\qw & & &\qw & \qw & \qw & \qw  & \qw & \qw & \qw & \rstick{c_1}\\
			\lstick{c_2} & \qw & \qw & \qw & \qw & \qw & \qw & & & &\qw & \qw & \qw & & & \qw  & \qw & \qw & & & \qw  & \qw & \ctrl{1} & \targ & \ctrl{1}  & \qw & \qw & & & \qw & \ctrl{7} & \qw & \qw & & & \qw & \ctrl{1} & \targ & \ctrl{1} & \qw & \qw & & &\qw & \qw & \qw & \qw  & \qw & \qw & \qw & \rstick{c_2}\\
			\lstick{c_3} & \qw & \qw & \qw & \qw & \qw & \qw & & \ar@{.}[dddd] & & \qw & \qw & \qw & & \ar@{.}[dddd] & \qw & \qw & \qw & & & \gate{H} & \gate{Z^\dag_2} & \gate{Z^\dag_2} & \qw & \gate{Z_2} & \gate{H} & \qw & & & \ctrl{6} & \qw & \qw & \qw & & & \gate{H}  & \gate{Z^\dag_2}  & \qw & \gate{Z_2} & \gate{Z_2} & \gate{H} & & &\qw \ar@{.}[ddd] & \qw & \qw & \qw  & \qw & \qw & \qw &\rstick{c_3}\\
			\cdot & \cdot & \cdot & \cdot & \cdot & \cdot & & & &\cdot & \cdot & \cdot & \cdot & \cdot & \cdot & \cdot & \cdot & \cdot & \cdot & \cdot & \cdot & \cdot & \cdot & \cdot & \cdot & \cdot & \cdot & \cdot & \cdot & \cdot & \cdot & \cdot & \cdot & \cdot & \cdot & \cdot & \cdot & \cdot & \cdot  & \cdot & \cdot & \cdot & \cdot & \cdot & \cdot & \cdot & \cdot \\
			\cdot & \cdot & \cdot & \cdot & \cdot & \cdot & & & & \cdot & \cdot & \cdot & \cdot & \cdot & \cdot & \cdot & \cdot & \cdot & \cdot & \cdot & \cdot & \cdot  & \cdot & \cdot & \cdot & \cdot & \cdot & \cdot & \cdot & \cdot & \cdot & \cdot & \cdot & \cdot & \cdot & \cdot & \cdot & \cdot & \cdot & \cdot & \cdot & \cdot & \cdot & \cdot & \cdot & \cdot & \cdot \\
			\lstick{c_{n-3}} & \qw & \qw & \qw & \qw & \ctrl{3} & \qw & & & & \qw & \qw & \qw & & & \qw & \qw & \qw & & & \qw & \qw & \qw & \qw & \qw & \qw & \qw & & & \qw & \qw & \qw & \qw & & & \qw & \qw & \qw & \qw & \qw & \qw & \qw & \qw & & & \ctrl{2} &\qw & \qw & \qw & \qw & \rstick{c_{n-3}}\\
			\lstick{c_{n-2}} & \qw & \qw & \qw & \ctrl{2} & \qw & \qw & & & & \qw & \qw & \qw & & & \qw & \ctrl{1} & \qw & & & \qw & \qw & \qw & \qw & \qw & \qw & \qw & & & \qw & \qw & \qw & \qw & & & \qw & \qw & \qw & \qw & \qw & \qw & \qw & \qw & & & \qw &\ctrl{1} & \qw & \qw & \rstick{c_{n-2}}\\
			\lstick{c_{n-1}} & \qw & \qw & \ctrl{1} & \qw & \qw & \qw & & & & \qw & \qw & \qw & & & \gate{H} & \gate{Z_2} & \qw & & & \qw & \qw & \qw & \qw & \qw & \qw & \qw & & & \qw & \qw & \qw & \qw & & & \qw & \qw & \qw & \qw & \qw & \qw & \qw & \qw & & & \gate{Z^\dag_4} &\gate{Z^\dag_2} & \gate{H} & \qw & \rstick{c_{n-1}}\\
			\lstick{c_n} & \qw & \gate{H} &\gate{Z_2} & \gate{Z_4} & \gate{Z_8} & \qw & & & & \qw & \gate{Z_{2^{n-1}}} & \qw & & & \qw & \qw & \qw & \qw & & & \cdot & \cdot & \cdot & \cdot & \cdot & \cdot & & & \gate{Z^\dag_{2^{n-2}}} & \gate{Z^\dag_{2^{n-1}}} & \gate{H} & \qw & & & \qw & \qw & \qw & \qw & \qw & \qw & \qw & \qw & & & \qw & \rstick{c_n\oplus {c_1}{c_2}....{c_{n-1}}} & & & & & & & &}
	}
	\caption{\em Equivalent circuit of $n$-MCT gate consisting of single control phase gates, CNOT and Hadamard gates. }
	\label{fig:intermidiate}
\end{figure*}
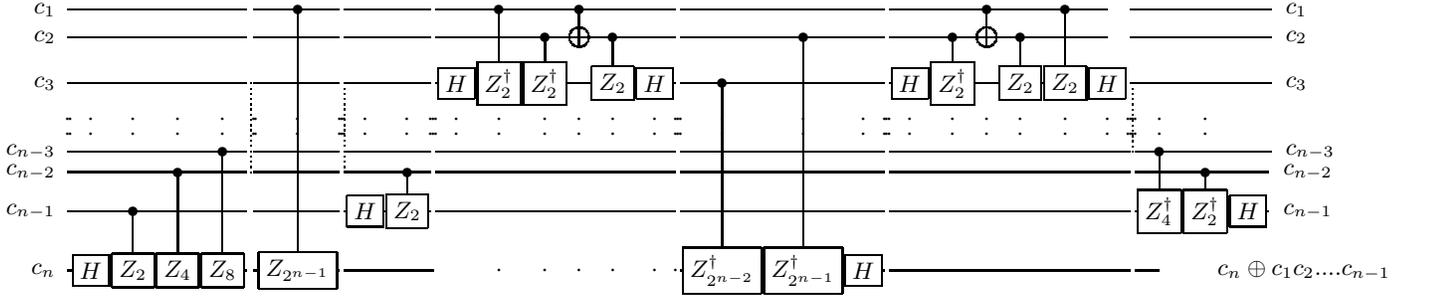

 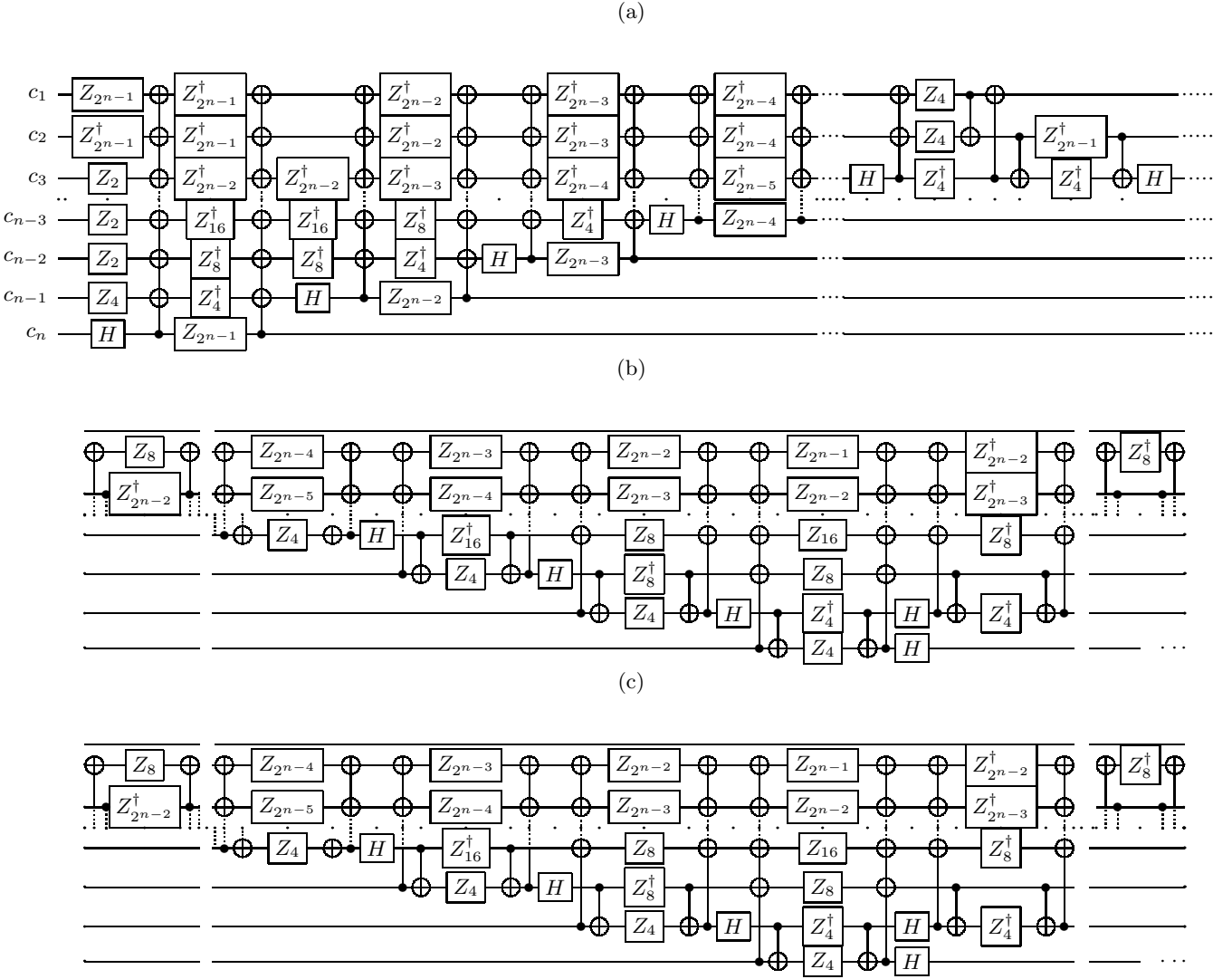
\begin{figure*}[t!]
    \centering
    \begin{subfigure}[t]{0.5\textwidth}
    	\caption{}
    	\label{subfig:C1}
     \circspace
	\centerline{
        \Qcircuit @C=0.3em @R=0.01em {
			\lstick{c_1} & \qw & \gate{Z_{2^{n-1}}} & \targ \qwx[2]& \gate{Z^\dag_{2^{n-1}}} & \targ \qwx[2] & \qw &  \targ \qwx[2] & \gate{Z^\dag_{2^{n-2}}}& \targ \qwx[2] & \qw & \targ \qwx[2] & \gate{Z^\dag_{2^{n-3}}} & \targ \qwx[2] & \qw & \targ \qwx[2] & \gate{Z^\dag_{2^{n-4}}} & \targ \qwx[2] &\qw & \cdot & \cdot & \cdot & \cdot & \qw & \targ & \gate{Z_4} & \ctrl{1} & \targ & \qw & \qw & \qw & \qw & \qw & \cdot & \cdot & \cdot & \cdot & \cdot \\
			\lstick{c_2} & \qw & \gate{Z^\dag_{2^{n-1}}}& \targ &  \gate{Z^\dag_{2^{n-1}}} & \targ & \qw & \targ & \gate{Z^\dag_{2^{n-2}}} & \targ & \qw & \targ & \gate{Z^\dag_{2^{n-3}}}  & \targ & \qw & \targ & \gate{Z^\dag_{2^{n-4}}}  & \targ & \qw & \cdot & \cdot & \cdot & \cdot & \qw & \targ & \gate{Z_4} & \targ & \qw & \ctrl{1} & \gate{Z^\dag_{2^{n-1}}} & \ctrl{1} & \qw & \qw & \cdot & \cdot & \cdot & \cdot & \cdot \\
			\lstick{c_3} & \qw & \gate{Z_2} & \targ \ar@{.}[ddd] & \gate{Z^\dag_{2^{n-2}}} & \targ \ar@{.}[ddd]& \gate{Z^\dag_{2^{n-2}}} & \targ \ar@{.}[ddd]  & \gate{Z^\dag_{2^{n-3}}}  & \targ \ar@{.}[ddd]   & \qw &\targ \ar@{.}[ddd] & \gate{Z^\dag_{2^{n-4}}} & \targ \ar@{.}[ddd] & \qw  &\targ \ar@{.}[ddd] & \gate{Z^\dag_{2^{n-5}}} & \targ \ar@{.}[ddd] & \qw & \cdot & \cdot & \cdot & \cdot & \gate{H} & \ctrl{-2} & \gate{Z^\dag_4} & \qw & \ctrl{-2} & \targ & \gate{Z^\dag_4} & \targ & \gate{H} & \qw & \qw & \cdot & \cdot & \cdot & \cdot  \\
			\cdot & \cdot & \cdot & \cdot & \cdot & \cdot & & & &\cdot & \cdot & \cdot & \cdot & \cdot & \cdot & \cdot & \cdot & \cdot & \cdot & \cdot & \cdot & \cdot & \cdot & \cdot & \cdot & \cdot & \cdot & \cdot & \cdot & \cdot & \cdot  \\
			\cdot & \cdot & \cdot & \cdot & \cdot & \cdot & & & & \cdot & \cdot & \cdot & \cdot & \cdot & \cdot & \cdot & \cdot & \cdot & \cdot & \cdot & \cdot  & \cdot & \cdot & \cdot & \cdot & \cdot & \cdot & \cdot & \cdot & \cdot & \cdot \\
			\lstick{c_{n-3}} & \qw & \gate{Z_2} & \targ & \gate{Z^\dag_{16}} & \targ & \gate{Z^\dag_{16}} & \targ & \gate{Z^\dag_8} & \targ & \qw & \targ & \gate{Z^\dag_4} & \targ & \gate{H} & \ctrl{0} & \gate{Z_{2^{n-4}}} & \ctrl{0} & \qw & \cdot & \cdot & \cdot & \cdot & \qw & \qw & \qw & \qw & \qw & \qw & \qw & \qw & \qw & \qw & \qw & \cdot & \cdot & \cdot & \cdot  \\
			\lstick{c_{n-2}} & \qw & \gate{Z_2} & \targ & \gate{Z^\dag_8} & \targ & \gate{Z^\dag_8} & \targ & \gate{Z^\dag_4} & \targ & \gate{H} & \ctrl{-1} & \gate{Z_{2^{n-3}}} & \ctrl{-1} & \qw & \qw & \qw & \qw & \qw & \cdot & \cdot & \cdot & \cdot & \qw & \qw & \qw & \qw & \qw & \qw & \qw & \qw & \qw & \qw & \qw & \cdot & \cdot & \cdot & \cdot \\
			\lstick{c_{n-1}} & \qw & \gate{Z_4} & \targ & \gate{Z^\dag_4} & \targ & \gate{H} & \ctrl{-2} & \gate{Z_{2^{n-2}}} & \ctrl{-2} & \qw & \qw & \qw & \qw & \qw & \qw & \qw & \qw & \qw & \cdot & \cdot & \cdot & \cdot & \qw & \qw & \qw & \qw & \qw & \qw & \qw & \qw & \qw & \qw& \qw & \cdot & \cdot & \cdot & \cdot \\
			\lstick{c_n} & \qw & \gate{H} & \ctrl{-3} & \gate{Z_{2^{n-1}}} & \ctrl{-3} & \qw & \qw & \qw & \qw & \qw & \qw & \qw & \qw & \qw & \qw & \qw & \qw & \qw & \cdot & \cdot & \cdot & \cdot & \qw & \qw & \qw & \qw & \qw & \qw & \qw & \qw & \qw & \qw & \qw & \cdot & \cdot & \cdot & \cdot}}

    \end{subfigure}%
    
    \begin{subfigure}[t]{0.5\textwidth}
    	\caption{}
    	\label{subfig:C2}
     \circspace
	\centerline{
       \Qcircuit @C=0.01em @R=0.01em {
	&\qw & \qw & \qw & \qw & \qw & \qw & \qw  & \qw & \push{\rule{.7em}{0em}} & \qw & \qw & \qw & \qw & \qw & \qw & \qw & \qw  & \qw & \qw & \qw & \qw & \qw & \qw & \qw & \qw & \qw  & \qw & \qw & \qw & \qw & \qw & \qw & \qw & \qw & \qw  & \qw & \qw & \qw  & \qw  & \qw  & \push{\rule{.7em}{0em}} & \qw & \qw & \qw & \qw & \qw & \qw & \qw & \qw  \\
	&\cdot & \cdot &  \targ \qwx[1] & \qw & \gate{Z_8} & \qw & \targ \qwx[1] & \qw & \push{\rule{.5em}{0em}} & \qw & \targ \qwx[1] & \qw & \gate{Z_{2^{n-4}}} & \qw & \targ \qwx[1]& \qw & \targ \qwx[1] & \qw & \gate{Z_{2^{n-3}}} & \qw & \targ \qwx[1]& \qw & \targ \qwx[1] & \qw & \gate{Z_{2^{n-2}}} & \qw & \targ \qwx[1]& \qw & \targ \qwx[1] & \qw & \gate{Z_{2^{n-1}}} & \qw & \targ \qwx[1]& \qw & \targ \qwx[1] & \qw & \gate{Z^\dag_{2^{n-2}}} & \qw & \targ \qwx[1]& \qw & & \push{\rule{.3em}{0em}}& \qw & \targ \qwx[1] & \qw & \gate{Z^\dag_8} & \qw & \targ \qwx[1] & \qw & \cdot & \cdot \\
	&\cdot & \cdot & \qw \ar@{.}[d]& \ctrl{0} \ar@{.}[d]& \gate{Z^\dag_{2^{n-2}}} & \qw & \ctrl{0} \ar@{.}[d] & \qw \ar@{.}[d] & \push{\rule{.5em}{0em}} & \qw & \targ \ar@{.}[ddd]& \qw &  \gate{Z_{2^{n-5}}} &  \qw & \targ \ar@{.}[ddd]& \qw & \targ \ar@{.}[ddd]& \qw &  \gate{Z_{2^{n-4}}} &  \qw & \targ \ar@{.}[ddd]& \qw & \targ \ar@{.}[ddd]& \qw &  \gate{Z_{2^{n-3}}} &  \qw & \targ \ar@{.}[ddd]& \qw & \targ \ar@{.}[ddd]& \qw &  \gate{Z_{2^{n-2}}} &  \qw & \targ \ar@{.}[ddd]& \qw & \targ \ar@{.}[ddd]& \qw &  \gate{Z^\dag_{2^{n-3}}} &  \qw & \targ \ar@{.}[ddd]& \qw & & \push{\rule{.3em}{0em}} & \qw & \qw \ar@{.}[d]& \ctrl{0} \ar@{.}[d]& \qw & \ctrl{0} \ar@{.}[d] & \qw \ar@{.}[d] & \qw & \cdot \\
	&\cdot & \cdot & \cdot & \cdot & \cdot & \cdot & \cdot & \cdot & \cdot & \cdot & \cdot & \cdot & \cdot & \cdot & \cdot & \cdot & \cdot & \cdot & \cdot & \cdot & \cdot & \cdot & \cdot & \cdot & \cdot & \cdot & \cdot & \cdot & \cdot & \cdot & \cdot & \cdot & \cdot & \cdot & \cdot & \cdot & \cdot & \cdot & \cdot & \cdot & \cdot & \cdot & \cdot & \cdot  & \cdot & \cdot & \cdot & \cdot & \cdot & \cdot & \cdot & \cdot & \cdot    \\
	&\cdot & \cdot & \cdot & \cdot & \cdot & \cdot & \cdot & \cdot & \cdot & \cdot\ar@{.}[d] & \cdot & \cdot \ar@{.}[d] & \cdot & \cdot & \cdot & \cdot  & \cdot & \cdot & \cdot & \cdot & \cdot & \cdot & \cdot & \cdot & \cdot & \cdot & \cdot & \cdot & \cdot & \cdot & \cdot & \cdot & \cdot & \cdot & \cdot & \cdot & \cdot & \cdot & \cdot & \cdot & \cdot & \cdot & \cdot & \cdot & \cdot & \cdot & \cdot & \cdot & \cdot & \cdot & \cdot & \cdot & \cdot    \\
	&\cdot & \cdot & \qw & \qw & \qw & \qw & \qw & \qw & \push{\rule{.5em}{0em}} & \qw & \ctrl{0} & \targ & \gate{Z_4} & \targ & \ctrl{0} & \gate{H} & \qw & \ctrl{1} & \gate{Z^\dag_{16}} & \ctrl{1} & \qw & \qw & \targ & \qw & \gate{Z_8} & \qw & \targ & \qw & \targ & \qw & \gate{Z_{16}} & \qw & \targ & \qw & \targ & \qw & \gate{Z^\dag_8} & \qw & \targ & \qw & \push{\rule{.7em}{0em}} & \qw & \qw & \qw & \qw & \qw & \qw & \qw & \qw & \qw & \qw & \cdot & \cdot & \cdot & \cdot  \\
	&\cdot & \cdot & \qw & \qw & \qw & \qw & \qw & \qw & \push{\rule{.5em}{0em}} & \qw & \qw & \qw & \qw & \qw & \qw & \qw & \ctrl{-1} & \targ & \gate{Z_4} & \targ & \ctrl{-1} & \gate{H} & \qw & \ctrl{1} & \gate{Z^\dag_{8}} & \ctrl{1} & \qw & \qw & \targ & \qw & \gate{Z_8} & \qw & \targ & \qw & \qw & \ctrl{1} & \qw & \ctrl{1} & \qw & \qw & \push{\rule{.7em}{0em}} & \qw & \qw & \qw & \qw & \qw & \qw & \qw & \qw & \qw & \qw & \qw & \cdot & \cdot & \cdot & \cdot \\
	&\cdot & \cdot & \qw & \qw & \qw & \qw & \qw & \qw & \push{\rule{.5em}{0em}} & \qw & \qw & \qw & \qw & \qw & \qw & \qw & \qw & \qw & \qw & \qw & \qw & \qw & \ctrl{-2} & \targ & \gate{Z_4} & \targ & \ctrl{-2} & \gate{H} & \qw & \ctrl{1} & \gate{Z^\dag_4} & \ctrl{1} & \qw & \gate{H} & \ctrl{-2} & \targ & \gate{Z^\dag_4} & \targ & \ctrl{-2} & \qw & \push{\rule{.7em}{0em}} & \qw & \qw & \qw & \qw & \qw & \qw & \qw & \qw & \qw& \qw & \cdot & \cdot & \cdot & \cdot \\
	&\cdot & \cdot & \qw & \qw & \qw & \qw & \qw & \qw & \push{\rule{.5em}{0em}} & \qw & \qw & \qw & \qw & \qw & \qw & \qw & \qw & \qw & \qw & \qw & \qw & \qw & \qw & \qw & \qw & \qw & \qw & \qw & \ctrl{-3} & \targ & \gate{Z_4} & \targ & \ctrl{-3} & \gate{H} & \qw & \qw & \qw & \qw & \qw & \qw & \push{\rule{.7em}{0em}} & \qw & \qw & \qw & \qw& \qw &\cdot & \cdot & \cdot }}
    \end{subfigure}
    \\
     \begin{subfigure}[t]{0.5\textwidth}
     	    	\caption{}
     	\label{subfig:C3}
    \circspace
	\centerline{
        \Qcircuit @C=0.01em @R=0.01em {
	&\qw & \qw & \qw & \qw & \qw & \qw & \qw  & \qw & \push{\rule{.7em}{0em}} & \qw & \qw & \qw & \qw & \qw & \qw & \qw & \qw  & \qw & \qw & \qw & \qw & \qw & \qw & \qw & \qw & \qw  & \qw & \qw & \qw & \qw & \qw & \qw & \qw & \qw & \qw  & \qw & \qw & \qw  & \qw  & \qw  & \push{\rule{.7em}{0em}} & \qw & \qw & \qw & \qw & \qw & \qw & \qw & \qw  \\
	&\cdot & \cdot &  \targ \qwx[1] & \qw & \gate{Z_8} & \qw & \targ \qwx[1] & \qw & \push{\rule{.5em}{0em}} & \qw & \targ \qwx[1] & \qw & \gate{Z_{2^{n-4}}} & \qw & \targ \qwx[1]& \qw & \targ \qwx[1] & \qw & \gate{Z_{2^{n-3}}} & \qw & \targ \qwx[1]& \qw & \targ \qwx[1] & \qw & \gate{Z_{2^{n-2}}} & \qw & \targ \qwx[1]& \qw & \targ \qwx[1] & \qw & \gate{Z_{2^{n-1}}} & \qw & \targ \qwx[1]& \qw & \targ \qwx[1] & \qw & \gate{Z^\dag_{2^{n-2}}} & \qw & \targ \qwx[1]& \qw & & \push{\rule{.3em}{0em}}& \qw & \targ \qwx[1] & \qw & \gate{Z^\dag_8} & \qw & \targ \qwx[1] & \qw & \cdot & \cdot \\
	&\cdot & \cdot & \qw \ar@{.}[d]& \ctrl{0} \ar@{.}[d]& \gate{Z^\dag_{2^{n-2}}} & \qw & \ctrl{0} \ar@{.}[d] & \qw \ar@{.}[d] & \push{\rule{.5em}{0em}} & \qw & \targ \ar@{.}[ddd]& \qw &  \gate{Z_{2^{n-5}}} &  \qw & \targ \ar@{.}[ddd]& \qw & \targ \ar@{.}[ddd]& \qw &  \gate{Z_{2^{n-4}}} &  \qw & \targ \ar@{.}[ddd]& \qw & \targ \ar@{.}[ddd]& \qw &  \gate{Z_{2^{n-3}}} &  \qw & \targ \ar@{.}[ddd]& \qw & \targ \ar@{.}[ddd]& \qw &  \gate{Z_{2^{n-2}}} &  \qw & \targ \ar@{.}[ddd]& \qw & \targ \ar@{.}[ddd]& \qw &  \gate{Z^\dag_{2^{n-3}}} &  \qw & \targ \ar@{.}[ddd]& \qw & & \push{\rule{.3em}{0em}} & \qw & \qw \ar@{.}[d]& \ctrl{0} \ar@{.}[d]& \qw & \ctrl{0} \ar@{.}[d] & \qw \ar@{.}[d] & \qw & \cdot \\
	&\cdot & \cdot & \cdot & \cdot & \cdot & \cdot & \cdot & \cdot & \cdot & \cdot & \cdot & \cdot & \cdot & \cdot & \cdot & \cdot & \cdot & \cdot & \cdot & \cdot & \cdot & \cdot & \cdot & \cdot & \cdot & \cdot & \cdot & \cdot & \cdot & \cdot & \cdot & \cdot & \cdot & \cdot & \cdot & \cdot & \cdot & \cdot & \cdot & \cdot & \cdot & \cdot & \cdot & \cdot  & \cdot & \cdot & \cdot & \cdot & \cdot & \cdot & \cdot & \cdot & \cdot    \\
	&\cdot & \cdot & \cdot & \cdot & \cdot & \cdot & \cdot & \cdot & \cdot & \cdot\ar@{.}[d] & \cdot & \cdot \ar@{.}[d] & \cdot & \cdot & \cdot & \cdot  & \cdot & \cdot & \cdot & \cdot & \cdot & \cdot & \cdot & \cdot & \cdot & \cdot & \cdot & \cdot & \cdot & \cdot & \cdot & \cdot & \cdot & \cdot & \cdot & \cdot & \cdot & \cdot & \cdot & \cdot & \cdot & \cdot & \cdot & \cdot & \cdot & \cdot & \cdot & \cdot & \cdot & \cdot & \cdot & \cdot & \cdot    \\
	&\cdot & \cdot & \qw & \qw & \qw & \qw & \qw & \qw & \push{\rule{.5em}{0em}} & \qw & \ctrl{0} & \targ & \gate{Z_4} & \targ & \ctrl{0} & \gate{H} & \qw & \ctrl{1} & \gate{Z^\dag_{16}} & \ctrl{1} & \qw & \qw & \targ & \qw & \gate{Z_8} & \qw & \targ & \qw & \targ & \qw & \gate{Z_{16}} & \qw & \targ & \qw & \targ & \qw & \gate{Z^\dag_8} & \qw & \targ & \qw & \push{\rule{.7em}{0em}} & \qw & \qw & \qw & \qw & \qw & \qw & \qw & \qw & \qw & \qw & \cdot & \cdot & \cdot & \cdot  \\
	&\cdot & \cdot & \qw & \qw & \qw & \qw & \qw & \qw & \push{\rule{.5em}{0em}} & \qw & \qw & \qw & \qw & \qw & \qw & \qw & \ctrl{-1} & \targ & \gate{Z_4} & \targ & \ctrl{-1} & \gate{H} & \qw & \ctrl{1} & \gate{Z^\dag_{8}} & \ctrl{1} & \qw & \qw & \targ & \qw & \gate{Z_8} & \qw & \targ & \qw & \qw & \ctrl{1} & \qw & \ctrl{1} & \qw & \qw & \push{\rule{.7em}{0em}} & \qw & \qw & \qw & \qw & \qw & \qw & \qw & \qw & \qw & \qw & \qw & \cdot & \cdot & \cdot & \cdot \\
	&\cdot & \cdot & \qw & \qw & \qw & \qw & \qw & \qw & \push{\rule{.5em}{0em}} & \qw & \qw & \qw & \qw & \qw & \qw & \qw & \qw & \qw & \qw & \qw & \qw & \qw & \ctrl{-2} & \targ & \gate{Z_4} & \targ & \ctrl{-2} & \gate{H} & \qw & \ctrl{1} & \gate{Z^\dag_4} & \ctrl{1} & \qw & \gate{H} & \ctrl{-2} & \targ & \gate{Z^\dag_4} & \targ & \ctrl{-2} & \qw & \push{\rule{.7em}{0em}} & \qw & \qw & \qw & \qw & \qw & \qw & \qw & \qw & \qw& \qw & \cdot & \cdot & \cdot & \cdot \\
	&\cdot & \cdot & \qw & \qw & \qw & \qw & \qw & \qw & \push{\rule{.5em}{0em}} & \qw & \qw & \qw & \qw & \qw & \qw & \qw & \qw & \qw & \qw & \qw & \qw & \qw & \qw & \qw & \qw & \qw & \qw & \qw & \ctrl{-3} & \targ & \gate{Z_4} & \targ & \ctrl{-3} & \gate{H} & \qw & \qw & \qw & \qw & \qw & \qw & \push{\rule{.7em}{0em}} & \qw & \qw & \qw & \qw& \qw &\cdot & \cdot & \cdot }}
    \end{subfigure}
    \caption{\em $n$-MCT gate decomposed into Clifford+$Z_N$-based fault tolerant circuit. The overall circuit is shown in three subfigures.}
	\label{fig:decompnMCT}
\end{figure*}

We establish the correctness of the decomposition using the proposed method. For ease of reading, the decomposed $n$-MCT circuit has been split into three parts  $P_1$, $P_2$ and $P_3$, as shown in Fig.~\ref{fig:decompnMCT}.
%
\begin{itemize}
\item [Case 1:] All control lines are set to 0, i.e., $c_i =0 $, \mbox{$1 \le i \le n-1$}. 

From Fig.~\ref{fig:lowermct}, it can be observed that all the conditional phase gates, CNOT-gates 
and $k$-MCT gates~($k<n$) are disabled. Subsequently, the net effect of a pair of H-gates is identity. Therefore, the value of the target qubit does not change.  

\item [Case 2:] One or more control lines~($c_i$) are assigned 0, whereas the remaining control lines are set to 1.

In this condition, all the MCT-gates as well as conditional gates associated with control lines $c_i$ will be disabled. As a result, the conditional gates followed by these disabled MCT-gates remain active but get canceled by opposite polarity active conditional gate associated with control bits other than $c_i$. Therefore, there is no effect on the target qubit under these settings of the control lines.
	
\item [Case 3:] All the control lines are set to 1. 

In Fig.~\ref{fig:lowermct}, each of the right most conditional phase gates over target line, followed by  the $k$-MCT gates~($k<n$) are disabled. However, all the conditional phase gates that independent to all these $k$-MCT gates are enabled. The
overall effect can be expressed as product of $Z_2.Z_4 \ldots Z_{2^{n-2}}$, which can be summarized as a geometric progression. 
	\begin{align}
	Z_{2^{n-2}}.\prod_{i=1}^{n-2} Z_{2^i}
	&= Z_{2^{n-2}}.\prod_{i=1}^{n-2} Z^{\frac{1}{2^i}} \nonumber \\
	&= Z^{\frac{1}{2}}.Z^{\frac{1}{2^2}} \ldots Z^{\frac{1}{2^{n-2}}}.Z^{\frac{1}{2^{n-2}}} \nonumber  \\
	&= Z^{\frac{1}{2}+{\frac{1}{2^2}}.......+{\frac{1}{2^{n-2}}}}.Z^{\frac{1}{2^{n-2}}} \nonumber  \\
	&= Z^{\sum_{i=1}^{n-2} a.r^{n-3}}.Z^{\frac{1}{2^{n-2}}} \nonumber  \\
	&\text{~~~~~~where } a=r={1/2} \nonumber \\
	&= Z^{\frac{{\frac{1}{2}}({1-{({\frac{1}{2})}^{n-2}}})}{1-{\frac{1}{2}}}}.Z^{\frac{1}{2^{n-2}}} \nonumber  \\
	&= Z^{1-\frac{1}{2^{n-2}}}.Z^{\frac{1}{2^{n-2}}} \nonumber  \\
	&= Z
	\end{align}
Therefore, the  overall effect of all the conditional gates acting over target line is equivalent to $n-1$-controlled Pauli's $Z$-gate, i.e., $n$-MCZ gate. The overall effect due to phase changes of $Z_N$-gates within a pair of $H$-gate which lies over target line in Fig.~\ref{fig:lowermct} is same as $n$-MCT gate.
\end{itemize}  

\begin{theorem}
	For $n \geq 4$, an $n$-MCT gate can be exactly decomposed over Clifford+$Z_N$-group without any ancilla.
\end{theorem}
\begin{proof}
	The proof follows from the construction presented above. 
\end{proof}


\subsection{Analysis of the phase count} 
\begin{theorem}[]For $n \geq 4$, an $n$-MCT  gate  can  be realized using Clifford+$Z_N$ library with  $\mathcal{O}(n^2)$ phase count without any ancilla.
	\label{ZN-count}
\end{theorem}
\begin{proof} 
To find phase count of the decomposed circuit in Fig.~\ref{fig:decompnMCT}, we compute the phase count of each block in Fig.~\ref{fig:decompnMCT} and finally sum up the individual phase counts for the entire circuit. 
 In Fig.~\ref{fig:decompnMCT}, some phase gates in each block of Fig.~\ref{fig:decompnMCT} are movable while other phase gates are immovable, i.e., they are shielded by a pair of CNOT gates. 

\noindent \textbf{Phase count of Fig.~\ref{subfig:C1}:} 
Observing the circuit from the input side,  the number of $Z_N$-gates that are movable are placed in two stacks.  The left most stack has $(n-1)$ number switchable $Z_N$-gates whereas {second stack contains $(n-4)$ number of $Z_N$-gates as the uppermost and bottom two lines do not contain any $Z_N$-gate. So, $(2n-5)$ $Z_N$-gates are present in these two stacks.} Furthermore, there are $(n-1)$ stacks of 
$Z_N$ gates boxed between a pair of CNOT gates. The number of gates in each stack decrements by 1 from $n$ $Z_N$ gates to two $Z_N$ gates in the last stack. Therefore, the total number of $Z_N$-gates can be computed as an arithmetic progression series with $a=2$ and $d=1$.
 \begin{align}
&\sum_{k=1}^{n-1}(a+(k-1)d) = {\frac{n.(n+1)-2}{2}={\frac{n^2}{2}+\frac{n}{2}-1}} \nonumber
 \end{align}
 The overall phase count~$ZC(\ref{subfig:C1})$ for the circuit in Fig.~\ref{subfig:C1} is given by
 \begin{equation}
ZC(\ref{subfig:C1}) = {\big(\frac{n^2}{2}+\frac{n}{{2}}-1\big)+(2n-5)={(\frac{n^2}{2}+\frac{5n}{2}-6})} 
\end{equation}
\noindent \textbf{Phase count of Fig.~\ref{subfig:C2}:} The $Z_N$ gates are in the form of an inverted pyramid structure. The count of phase gates in the $Z_N$ gate stacks increases from left to the mid position and then decreases. In the first half, the count of phase gates starts from $3$ number  and increases till $(n-1)$ gates, contributing ($\frac{n^2}{2}-\frac{n}{2}-3$) as phase depth. For the second half, the count of phase gate starts from $(n-3)$ and decreases to $2$, which adds phase count of $\frac{n^2}{2}-\frac{5n}{2}+2$. Therefore, the phase count~$ZC(\ref{subfig:C2})$ for the circuit in Fig.~\ref{subfig:C2} is given by
\begin{equation}
ZC(\ref{subfig:C2}) = \big(\frac{n^2}{2}-\frac{n}{2}-3 \big) + \big(\frac{n^2}{2}-\frac{5n}{2}+2 \big) {= n^2-3n-1}
\end{equation}
\noindent \textbf{Phase count of Fig.~\ref{subfig:C3}:} Similar 
to the previous two sub-circuits, the $Z_N$-gates are placed in two ways i.e. movable and immovable. There are three stacks of movable $Z_N$-gates --- the first stack has $n-3$ gates, the second stack has $(n-5)$ gates and the third one has a single $Z_N$ gate. Thus, the number of movable $Z_N$-gates in this 
sub-circuit is $(2n-7)$. On other hand, this sub-circuit contains stack of immovable $Z_N$-gates 
starting with $3$ $Z_N$ gates and increasing till $n-1$ gates and there is an additional immovable $Z_N$ gate. Thus, the number of immovable $Z_N$ gates  is 
\begin{equation*}
    (\frac{n^2}{2}-\frac{n}{2}-2)
\end{equation*}
The phase count~$ZC(\ref{subfig:C3})$ for the circuit in Fig.~\ref{subfig:C3} is 
\begin{equation}
ZC(\ref{subfig:C3}) = (2n-7) + (\frac{n^2}{2}-\frac{n}{2}-2) = (\frac{n^2}{2}+\frac{3n}{2}-9)
\end{equation}
\noindent \textbf{Phase count of $\mathbf{n}$-MCT:}
By summing up the phase count of the individual subcircuits in Fig.~\ref{fig:decompnMCT}, we obtain the phase count~$ZC(n\text{-MCT})$ of the Clifford+$Z_n$ fault tolerant realization of $n$-MCT gate. 
 \begin{align} 
&~~ZC(n\text{-MCT}) \nonumber \\
&= ZC(\ref{subfig:C1}) + ZC(\ref{subfig:C2}) + ZC(\ref{subfig:C3})  \nonumber \\
&= {(\frac{n^2}{2}+\frac{5n}{2}-6})+(n^2-3n-1)+(\frac{n^2}{2}+\frac{3n}{2}-9) \nonumber\\
&= 2n^2+n-16 = \mathcal{O}(n^2)
 \end{align}  

\end{proof}

\subsection{Analysis of phase depth} 
\begin{theorem} {An $n$-MCT  gate can be decomposed with $\mathcal{O}(n)$  phase depth over  Clifford+$Z_N$-group without any ancilla.}
\label{ZN-depth}
\end{theorem}

\begin{proof} We use an approach similar to the proof 
of theorem~~\ref{ZN-count} to prove this theorem.

\noindent \textbf{Phase depth of Fig.~\ref{subfig:C1}:} There are two stacks of movable type of $Z_N$-gates, that require two cycles for execution. On other hand, the immovable $Z_N$-gates are placed in $(n-1)$ stacks, that would require $(n-1)$ cycles for execution. So, the phase depth~$ZD(\ref{subfig:C1})$ of the Fig.~\ref{subfig:C1} is computed as follows. 
\begin{equation}
ZD(\ref{subfig:C1})=\left\{
	\begin{array}{ll}
	4 & n=4  \\
	(n+1) & n\geq 5\\
	\end{array} 
	\right. 
	\end{equation}
\noindent \textbf{Phase depth of Fig.~\ref{subfig:C2}:}
In Fig.~\ref{subfig:C2}, there are $(n-3)$ stacks with increasing  number
of $Z_N$ gates and $(n-4)$ stacks with decreasing number
of $Z_N$ gates. Therefore, the phase depth~$ZD(\ref{subfig:C2})$ of Fig.~\ref{subfig:C2} is given by :
\begin{equation}
    ZD(\ref{subfig:C2}) = 2n - 7
\end{equation}
\noindent \textbf{Phase depth of Fig.~\ref{subfig:C3}:}
There are $(n-2)$ stacks of immovable  $Z_N$ gates.  One movable $Z_N$ gate
can be executed in parallel with either stack of immovable $Z_N$ gates and hence this does not contribute to additional phase depth. In addition, there are two remaining stacks of movable $Z_N$ gates, with all gates in each stack being executed in parallel. Thus, the phase depth~$ZD(\ref{subfig:C3})$ of this sub-circuit is given by :
\begin{equation}
ZD(\ref{subfig:C3})=\left\{
	\begin{array}{ll}
	2 & n=4  \\
	4 & n=5 \\
	n & n\geq 6\\
	\end{array} 
	\right.
	\end{equation}
\noindent \textbf{Phase depth of $\mathbf{n}$-MCT:}
The phase depth~$ZD(n\text{-MCT})$ of the proposed
decomposition of $n$-MCT gate using Clifford+$Z_N$ library as shown in Fig.~\ref{fig:decompnMCT} is computed as follows.
 \begin{align} 
&~~ZD(n\text{-MCT}) \nonumber \\
&= ZD(\ref{subfig:C1}) + ZD(\ref{subfig:C2}) + ZD(\ref{subfig:C3})  \nonumber \\
&=\left\{
	\begin{array}{ll}
	7 & n=4  \\
	13 & n=5 \\
	4n-6 & n\geq 6\\
	\end{array} 
	\right.
&= \mathcal{O}(n) \label{ZN-finaldepth}
\end{align}

\end{proof}

\section{Design of unit phase depth based structure for $n$-MCT gate}\label{sec:unitph}
\noindent Quantum circuits with low phase depth are important~\cite{Bravyi}. In this regard, we present a unit phase depth decomposition of $n$-MCT gate, at the cost of additional ancillary lines. 
\begin{figure*}[ht]
	\circspace
	\centerline{
		\Qcircuit @C=0.25em @R=0.1em {
			\lstick{c_1} & \qw & \qw & \ctrl{8} & \qw & \qw & \qw & \qw & \qw & \qw & \ctrl{7} & \qw & \qw & \qw & \qw & \qw & \qw & \ctrl{11} & \ctrl{13} & \qw & \qw & \qw & \qw & \qw & \qw  & \gate{Z_8} & \qw & \qw & \qw & \qw & \qw & \qw  & \ctrl{13} & \ctrl{11} & \qw & \qw & \qw & \qw & \qw & \qw & \ctrl{7}& \qw & \qw & \qw & \qw & \qw & \qw & \ctrl{8} & \qw & \qw & \rstick{c_1} \\ 
			\lstick{c_2} & \qw & \qw & \qw & \ctrl{8} & \qw & \qw & \qw & \qw & \ctrl{7} & \qw & \ctrl{5} & \qw & \qw & \qw & \qw & \ctrl{11} & \qw & \qw & \qw & \qw & \qw & \qw & \qw & \qw & \gate{Z_8} & \qw & \qw & \qw & \qw & \qw & \qw & \qw & \qw  & \ctrl{11} & \qw & \qw & \qw & \qw & \ctrl{5} & \qw & \ctrl{7} & \qw & \qw & \qw & \qw & \ctrl{8} & \qw & \qw & \qw & \rstick{c_2}\\ 
			\lstick{c_3} & \qw & \qw & \qw & \qw & \ctrl{8} & \qw & \qw & \ctrl{7} & \qw & \qw & \qw & \ctrl{3} & \qw & \qw & \ctrl{11} & \qw & \qw & \qw & \qw & \qw & \qw & \qw & \qw & \qw & \gate{Z_8}  & \qw & \qw & \qw & \qw & \qw & \qw & \qw & \qw & \qw & \ctrl{11} & \qw & \qw & \ctrl{3} & \qw & \qw & \qw & \ctrl{7} & \qw & \qw & \ctrl{8}  & \qw & \qw & \qw & \qw & \rstick{c_3} \\ 
			\lstick{t_{in}} & \qw & \gate{H} & \qw & \qw & \qw & \ctrl{8} & \ctrl{7} & \qw & \qw & \qw & \qw & \qw & \ctrl{1} & \ctrl{9} & \qw & \qw & \qw & \qw & \qw & \qw & \qw & \qw & \qw & \qw & \gate{Z_8} & \qw & \qw & \qw & \qw & \qw & \qw & \qw & \qw & \qw & \qw & \ctrl{9} & \ctrl{1} & \qw & \qw & \qw & \qw & \qw & \ctrl{7} & \ctrl{8}  & \qw & \qw & \qw & \gate{H} & \qw & \rstick{t_{out}}\\ 
			\lstick{\ket{0}} & \qw & \qw & \qw & \qw & \qw & \qw & \qw & \qw & \qw & \qw & \qw & \qw & \targ & \qw & \qw & \qw & \qw & \qw & \targ & \qw & \qw & \qw & \qw & \qw & \gate{Z_8} & \qw & \qw & \qw & \qw & \qw & \targ & \qw & \qw & \qw & \qw & \qw & \targ  & \qw & \qw & \qw & \qw & \qw & \qw & \qw & \qw & \qw & \qw & \qw & \qw & \rstick{\ket{0}}\\ 
			\lstick{\ket{0}} & \qw & \qw & \qw & \qw & \qw & \qw & \qw & \qw & \qw & \qw & \qw & \targ & \qw & \qw & \qw & \qw & \qw & \qw & \qw & \targ & \qw & \qw & \qw & \qw & \gate{Z_8} & \qw & \qw & \qw & \qw & \targ & \qw & \qw & \qw & \qw & \qw & \qw & \qw & \targ & \qw & \qw & \qw & \qw & \qw & \qw & \qw & \qw & \qw & \qw & \qw & \rstick{\ket{0}}\\ 
			\lstick{\ket{0}} & \qw  & \qw  & \qw  & \qw  & \qw   & \qw & \qw & \qw & \qw  & \qw & \targ & \qw & \qw & \qw & \qw & \qw  & \qw & \qw & \qw  & \qw & \targ & \qw  & \qw  & \qw & \gate{Z_8}  & \qw  & \qw  & \qw & \targ & \qw & \qw & \qw & \qw & \qw  & \qw & \qw & \qw  & \qw & \targ  & \qw  & \qw  & \qw  & \qw  & \qw   & \qw & \qw & \qw & \qw  & \qw & \rstick{\ket{0}}\\ 
			\lstick{\ket{0}} & \qw  & \qw & \qw & \qw & \qw  & \qw  & \qw  & \qw  & \qw  & \targ & \qw & \qw  & \qw  & \qw  & \qw  & \qw  & \qw  & \qw & \qw & \qw & \qw & \targ  & \qw  & \qw & \gate{Z_8} & \qw & \qw & \targ  & \qw & \qw  & \qw  & \qw  & \qw  & \qw  & \qw  & \qw & \qw & \qw & \qw & \targ & \qw  & \qw & \qw & \qw & \qw  & \qw  & \qw  & \qw  & \qw & \rstick{\ket{0}}\\  
			\lstick{\ket{0}} & \qw & \qw & \targ & \qw & \qw & \qw & \qw & \qw & \targ & \qw & \qw & \qw & \qw & \qw & \qw & \qw & \qw & \qw & \qw & \qw & \qw & \qw & \qw & \qw & \gate{Z^\dag_8}  & \qw & \qw & \qw & \qw & \qw & \qw & \qw & \qw & \qw & \qw & \qw & \qw & \qw & \qw & \qw & \targ & \qw & \qw & \qw & \qw & \qw & \targ & \qw & \qw & \rstick{\ket{0}}\\ 
			\lstick{\ket{0}} & \qw & \qw & \qw & \targ & \qw & \qw & \qw & \targ & \qw & \qw  & \qw & \qw & \qw & \qw & \qw & \qw & \qw & \qw & \qw & \qw & \qw & \qw & \ctrl{5} & \qw & \gate{Z^\dag_8} & \qw & \ctrl{5}  & \qw & \qw  & \qw & \qw & \qw & \qw & \qw & \qw & \qw & \qw & \qw & \qw & \qw & \qw & \targ & \qw & \qw & \qw & \targ & \qw & \qw & \qw & \rstick{\ket{0}}\\ 
			\lstick{\ket{0}} & \qw & \qw & \qw & \qw & \targ & \qw & \targ & \qw & \qw & \qw & \qw & \qw & \qw & \qw & \qw & \qw & \qw & \qw & \qw & \qw & \qw & \ctrl{-3} & \qw & \qw & \gate{Z^\dag_8} & \qw & \qw & \ctrl{-3} & \qw & \qw & \qw & \qw & \qw & \qw & \qw & \qw & \qw & \qw & \qw & \qw & \qw & \qw & \targ & \qw & \targ  & \qw & \qw & \qw & \qw & \rstick{\ket{0}}\\ 
			\lstick{\ket{0}} & \qw & \qw & \qw & \qw & \qw & \targ & \qw & \qw & \qw & \qw & \qw & \qw & \qw & \qw & \qw & \qw & \targ & \qw & \qw & \qw & \ctrl{-5} & \qw & \qw & \qw & \gate{Z^\dag_8} & \qw & \qw & \qw & \ctrl{-5} & \qw & \qw & \qw & \targ & \qw & \qw & \qw & \qw & \qw & \qw & \qw & \qw & \qw & \qw & \targ & \qw & \qw & \qw & \qw & \qw & \rstick{\ket{0}}\\ 
			\lstick{\ket{0}} & \qw & \qw & \qw & \qw & \qw & \qw & \qw & \qw & \qw & \qw & \qw & \qw & \qw & \targ & \qw & \targ & \qw & \qw & \qw & \ctrl{-7} & \qw  & \qw & \qw & \ctrl{2} & \gate{Z^\dag_8} & \ctrl{2} & \qw  & \qw & \qw & \ctrl{-7} & \qw & \qw & \qw & \targ & \qw & \targ & \qw & \qw & \qw & \qw & \qw & \qw & \qw & \qw & \qw & \qw & \qw & \qw & \qw & \rstick{\ket{0}}\\ 
			\lstick{\ket{0}} & \qw & \qw & \qw & \qw & \qw & \qw & \qw & \qw & \qw & \qw & \qw & \qw & \qw & \qw & \targ & \qw & \qw & \targ & \ctrl{-9} & \qw & \qw & \qw & \qw & \qw & \gate{Z^\dag_8} & \qw & \qw & \qw & \qw & \qw & \ctrl{-9} & \targ & \qw & \qw & \targ & \qw & \qw & \qw & \qw & \qw & \qw & \qw & \qw & \qw & \qw & \qw & \qw & \qw & \qw & \rstick{\ket{0}}\\ 
			\lstick{\ket{0}} & \qw & \qw & \qw & \qw & \qw & \qw& \qw & \qw & \qw & \qw & \qw & \qw & \qw & \qw & \qw & \qw & \qw & \qw & \qw & \qw & \qw & \qw & \targ & \targ & \gate{Z^\dag_8} & \targ & \targ & \qw & \qw & \qw & \qw & \qw & \qw& \qw & \qw & \qw & \qw & \qw & \qw & \qw & \qw & \qw & \qw & \qw & \qw & \qw & \qw & \qw & \qw & \rstick{\ket{0}} 
	}}
	\caption{\em Fault tolerant unit phase depth decomposition of $4$-MCT gate using Clifford+$Z_N$ library.}
	\label{fig:onemore}
	\vspace{-0.2cm}
\end{figure*}
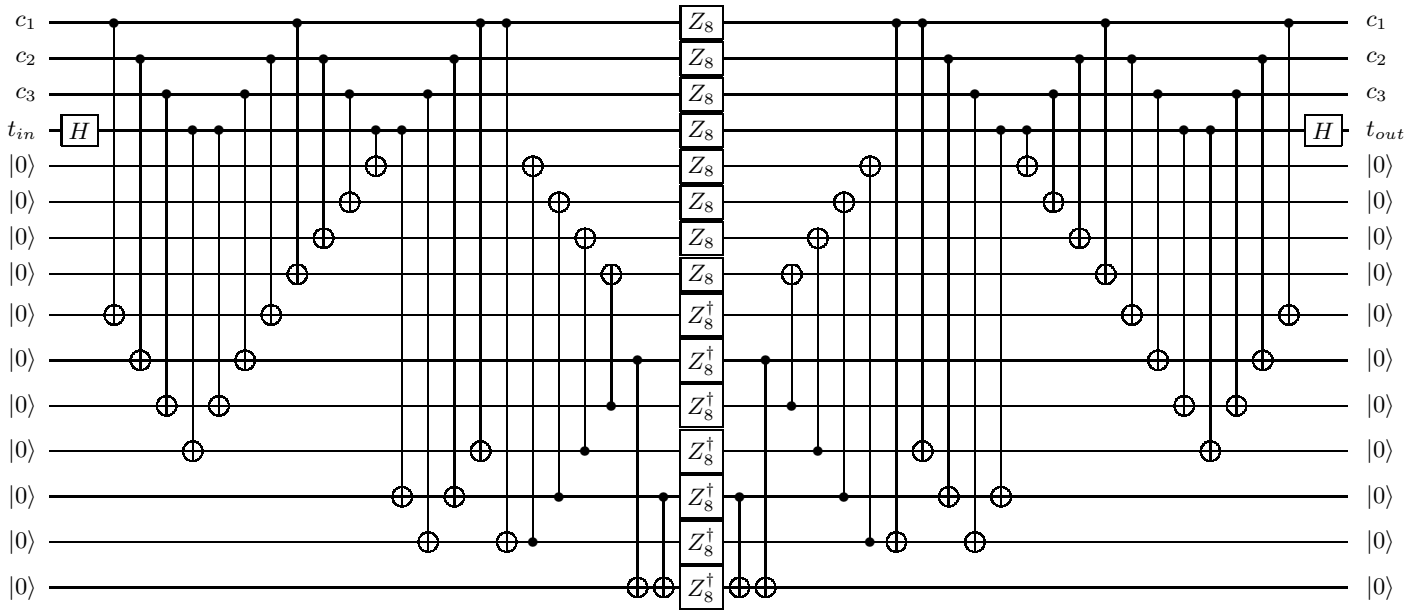
\begin{theorem} 
For any $n \geq 4$, an $n$-MCT gate can be simulated by Clifford+$Z_N$-group in unit phase depth, at the cost of $\mathcal{O}(2^n)$ phase count and $\mathcal{O}(2^n)$ ancilla lines.
\end{theorem}


\begin{proof} 
 Prior to verifying this theorem, we provide the mathematical representation of two circuit structures, as shown in Fig.~\ref{fig:mathcirc}. 
\begin{figure}[h]
\begin{subfigure}[t]{0.4\columnwidth}
	\circspace
\caption{}
\label{fig:small}
	\centerline{
\Qcircuit @C=.8em @R=.7em {\lstick{\ket{x}} & \qw &\gate{T} & \qw & \rstick{T\ket{x}}}}
\end{subfigure}
\begin{subfigure}[t]{0.5\columnwidth}
\circspace
\caption{}
\label{fig:large}
\centerline{
\Qcircuit @C=.8em @R=.7em { \lstick{\ket{x}} & \ctrl{1} &\qw & \ctrl{1} & \qw & \rstick{\ket{x}}\\
	  \lstick{\ket{y}} & \targ & \gate{T} & \targ & \qw & \rstick{T^{x\oplus y} \ket{y}}
}}
\end{subfigure}
\caption{\em $T^x\ket{x}$ and $T^{x\oplus y} \ket{xy}$ correspond to the mathematical form of circuits (\subref{fig:small}) and (\subref{fig:large})  respectively. }
\label{fig:mathcirc}
\end{figure}
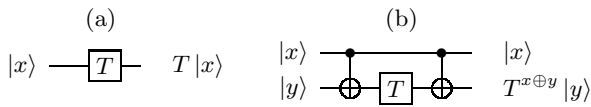

Without loss of generality, we provide a constructive proof for realizing a unit phase depth fault tolerant circuit for $4$-MCT.
The relation between input and output qubits  of  the $4$-MCT gate is based on Kronecker product of unitary matrix, as shown in \eqref{eq:4mct}.
\begin{equation}\label{eq:4mct}
\ket{c_1 c_2 c_3 t_{out}} = (I^{\otimes 3} \otimes H)(C^3(Z))(I^{\otimes 3} \otimes H) \ket{c_1 c_2 c_3 t_{in}}
\end{equation}
where $t_{out}={t_{in}\oplus c_{1}c_{2}c_{3}}$ and $t_{in}=t$. We consider the decomposition of $4$-MCT into pair of H gates and $4$-MCZ gate as shown in Fig.~\ref{subfig:n-MCZ} and then proceed to decompose the $4$-MCZ gate into a circuit with control free phase gates with unit phase depth. For simplicity of representation, we consider $a$, $b$, $c$, $d$ as inputs and $p$, $q$, $r$, $s$ as outputs of the $4$-MCZ gate. The relationship between the input and the outputs can be expressed by \eqref{eq:pqrs}.
\begin{align}
\ket{pqrs} &= (-1)^{abcd}\ket{abcd} \nonumber \\
 &= (Z)^{abcd}\ket{abcd} \nonumber \\ 
 &= (Z)^\frac{8abcd}{8}\ket{abcd}  \label{eq:pqrs}
\end{align}
To decompose unitary operators multiple control inputs 
in terms of elementary gates, the exponentiation of the minterm~($8abcd$) has to be decomposed into exponentiation of single qubit variables. This is achieved by the modulo-2 inclusion and exclusion rule~\cite{Biswal}. $+$ and $-$ symbols represent the usual arithmetic operator while $\oplus$ represents XOR.
\begin{align}
8abcd &= a+b+c+d- a\oplus b-a\oplus c \nonumber\\
&-a\oplus d -b\oplus c 
- b\oplus d -c\oplus d \nonumber\\ 
&+ a\oplus b \oplus c + a\oplus b \oplus d  + a\oplus c \oplus d  \nonumber \\
&+ b\oplus c \oplus d - a\oplus b \oplus c \oplus d 
\label{moduloexp}
\end{align}
Substituting~\eqref{moduloexp} in \eqref{eq:pqrs}, we get the control free phase get decomposition for $4$-MCZ gate. 
\begin{align}
\ket{pqrs}&=Z^a_8Z^b_8Z^c_8Z^d_8{Z^\dag}^{a\oplus b}_8{Z^\dag}^{a\oplus c}_8{Z^\dag}^{a\oplus d}_8{Z^\dag}^{b\oplus c}_8{Z^\dag}^{b\oplus d}_8{Z^\dag}^{c\oplus d}_8 \nonumber \\
&Z^{a\oplus b \oplus c}_8Z^{a\oplus b \oplus d}_8Z^{a\oplus c \oplus d}_8Z^{b\oplus c \oplus d}_8{Z^\dag}^{a\oplus b \oplus c \oplus d}_8\ket{abcd}
\label{mathexp}
\end{align}

The circuit corresponding to \eqref{mathexp} can be be 
obtained by using expanding the expression similar to the technique presented in Fig.~\ref{fig:mathcirc}, with the ancillary inputs initialized to $\ket{0}$. 
Fig.~\ref{fig:onemore} presents the equivalent quantum circuit of Fig.~\ref{fig:4mct} in the form of transversal primitive unitary operators with unit phase-depth. It is straightforward to extend the proposed approach  for $n$-MCT gate, given any arbitrary $n\ge 4$. 

In Fig.~\ref{fig:onemore}, all the phase gates  can be executed in one cycle which implies that the constructed fault-tolerant circuit has unit phase depth.

\noindent \textbf{Analysis of phase count:} The decomposition of $4$-MCT in Fig.~\ref{fig:onemore} has $8$ $Z_8$-gates~($2^3-1$) and equal number of $Z^\dag_8$ gates. We can infer the phase count by extrapolating the proposed decomposition approach. For any arbitrary $n \ge 4$, the phase count of the circuit is equal to $(2^{n}-1)$.

\noindent \textbf{Analysis of ancilla count:} The unit phase depth of the decomposed circuit is feasible due to parallel operations of all the $15$ phase gates, thereby incurring the need for additional $11$~($(2^4-4-1)$) ancillary lines. For decomposing $n$-MCT circuit using Clifford+$Z_N$ library with unit phase depth, $(2^{n}-n-1)$ ancillary lines are required.
\end{proof}



%
%
%
%
%

\section{Related works}\label{sec:exp}
\noindent The idea of decomposition of $n$-MCT gates into cascade of Toffolli gates, using ancillary lines was proposed by Barrenco et al.~\cite{Barrenco}. 
Furthermore, a decomposition technique for Toffoli gate
into fault-tolerant NCV quantum gate library was  proposed as well. The 
 NCV quantum gate library consists  of NOT, CNOT , $CV$ and $CV^\dag$ gates.
Based on \cite[Lemma-~7.3]{Barrenco}, a novel decomposition approach using NCV gate library was reported in \cite{Miller}. Maslov et. al proposed the decomposition of $n$-MCT gate using Peres gates, in the presence of ancillary inputs~\cite{Maslov} to reduce the quantum costs.
Abdessaied et al. proposed Clifford+$T$-based decomposition of $n$-MCT gate
using a single ancillary line, where the $T$-depth of the decomposed circuit is equal to $8(c-2)$, where $c$ is the number of control lines~($c \geq 5$)~\cite{Abdessaied}.
In~\cite{Saeedi}, the decomposition approach for 
$n$-MCT gates was proposed using controlled rotation gate, in linear depth using no ancillary lines~\cite{Saeedi}. 
 However, these approaches do not provide fault tolerant
 decomposition of $n$-MCT, except~\cite{Abdessaied}.  However, the approach proposed in \cite{Abdessaied} uses an ancilla input. Therefore, these approaches cannot by directly compared to the work proposed in this work.
 Our approach realizes $n$-MCT gate with phase depth roughly 50\% of the proposed approach in~\cite{Abdessaied}, without using any ancillary inputs.

\section{Conclusion}\label{sec:conc}
\noindent Universal fault-tolerant quantum computation needs error-free operations consisting of long sequence of transversal primitive operators. As conditional Toffoli gates form the building block for quantum algorithms, it is important to 
decompose these gates into a fault tolerant structure.  


To achieve this goal, we have proposed a new fault tolerant quantum gate library, namely the Clifford+$Z_N$ library.  Using the proposed library, 
we have presented a novel ancilla free decomposition algorithm for $n$-MCT into fault tolerant circuit with linear phase depth using Clifford+$Z_N$ library.

Due to the low coherence time,  it is necessary to run all operations prior to qubit decoherence. This entails the need for quantum circuits with low phase depth, as the phase depth determines the latency of the circuit. We have proposed a novel approach to decompose $n$-MCT into a fault-tolerant circuit with unit phase depth. This approach allows the phase depth of the proposed fault tolerant circuit to remain constant, invariant of the number of control lines in the MCT gate. The proposed approach paves a path for scalable shallow phase depth circuit construction to realize quantum algorithms.

\bibliographystyle{unsrt}
\bibliography{ref}

\begin{thebibliography}{10}

\bibitem{Shor}
Peter~W Shor.
\newblock Polynomial-time algorithms for prime factorization and discrete
  logarithms on a quantum computer.
\newblock {\em SIAM review}, 41(2):303--332, 1999.

\bibitem{Grover}
Lov~K Grover.
\newblock Quantum mechanics helps in searching for a needle in a haystack.
\newblock {\em Physical review letters}, 79(2):325, 1997.

\bibitem{Buzek}
Vladimir Bu{\v{z}}ek and Mark Hillery.
\newblock Quantum copying: Beyond the no-cloning theorem.
\newblock {\em Physical Review A}, 54(3):1844, 1996.

\bibitem{Ladd}
Thaddeus~D Ladd, Fedor Jelezko, Raymond Laflamme, Yasunobu Nakamura,
  Christopher Monroe, and Jeremy~Lloyd O’Brien.
\newblock Quantum computers.
\newblock {\em Nature}, 464(7285):45, 2010.

\bibitem{Steane}
A.~Steane.
\newblock The ion trap quantum information processor.
\newblock {\em Applied Physics B}, 64(5):623--643, 1997.

\bibitem{Paetznic}
Adam Paetznick and Ben~W Reichardt.
\newblock Universal fault-tolerant quantum computation with only transversal
  gates and error correction.
\newblock {\em Physical review letters}, 111(9):090505, 2013.

\bibitem{Cody_Jones_2012}
N~Cody Jones, James~D Whitfield, Peter~L McMahon, Man-Hong Yung, Rodney~Van
  Meter, Al{\'{a}}n Aspuru-Guzik, and Yoshihisa Yamamoto.
\newblock Faster quantum chemistry simulation on fault-tolerant quantum
  computers.
\newblock {\em New Journal of Physics}, 14(11):115023, 2012.

\bibitem{Miller}
D~Michael Miller, Robert Wille, and Zahra Sasanian.
\newblock Elementary quantum gate realizations for multiple-control toffoli
  gates.
\newblock In {\em 2011 41st IEEE International Symposium on Multiple-Valued
  Logic}, pages 288--293. IEEE, 2011.

\bibitem{Barrenco}
Adriano Barenco, Charles~H Bennett, Richard Cleve, David~P DiVincenzo, Norman
  Margolus, Peter Shor, Tycho Sleator, John~A Smolin, and Harald Weinfurter.
\newblock Elementary gates for quantum computation.
\newblock {\em Physical review A}, 52(5):3457, 1995.

\bibitem{Saeedi}
Mehdi Saeedi and Massoud Pedram.
\newblock Linear-depth quantum circuits for n-qubit toffoli gates with no
  ancilla.
\newblock {\em Physical Review A}, 87(6):062318, 2013.

\bibitem{Eastin}
Bryan Eastin and Emanuel Knill.
\newblock Restrictions on transversal encoded quantum gate sets.
\newblock {\em Phys. Rev. Lett.}, 102, 2009.

\bibitem{Fowler-1}
Austin~G. Fowler, Matteo Mariantoni, John~M. Martinis, and Andrew~N. Cleland.
\newblock Surface codes: Towards practical large-scale quantum computation.
\newblock {\em Phys. Rev. A}, 86:032324, 2012.

\bibitem{Martinis}
John~M. Martinis.
\newblock Qubit metrology for building a fault-tolerant quantum computer.
\newblock {\em Npj Quantum Information}, 1:032324, 2015.

\bibitem{Bravyi}
Sergey Bravyi, David Gosset, and Robert Koenig.
\newblock Quantum advantage with shallow circuits.
\newblock {\em Science}, 362(6412):308--311, 2018.

\bibitem{Gottesman}
D.~{Gottesman}.
\newblock Stabilizer codes and quantum error correction.
\newblock In {\em PhD.~thesis}. Caltech, 1997.

\bibitem{Mishra}
Prashant Mishra and Austin Fowler.
\newblock Resource comparison of two surface code implementations of small
  angle z rotations.
\newblock {\em arXiv preprint arXiv:1406.4948}, 2014.

\bibitem{Landahl}
Andrew~J Landahl and Chris Cesare.
\newblock Complex instruction set computing architecture for performing
  accurate quantum z rotations with less magic.
\newblock {\em arXiv preprint arXiv:1302.3240}, 2013.

\bibitem{Fowler-Stephens}
A~M.Stephens A~G.~Fowler and P~Groszkowski.
\newblock High-threshold universal quantum computation on the surface code.
\newblock {\em Phys. Rev. A}, 80:052312, 2009.

\bibitem{Amy-Maslov}
M.~{Amy}, D.~{Maslov}, M.~{Mosca}, and M.~{Roetteler}.
\newblock A meet-in-the-middle algorithm for fast synthesis of depth-optimal
  quantum circuits.
\newblock {\em IEEE Transactions on Computer-Aided Design of Integrated
  Circuits and Systems}, 32(6):818--830, 2013.

\bibitem{Soeken}
Mathias Soeken, D~Michael Miller, and Rolf Drechsler.
\newblock Quantum circuits employing roots of the pauli matrices.
\newblock {\em Physical Review A}, 88(4):042322, 2013.

\bibitem{Biswal}
Laxmidhar Biswal, Rakesh Das, Chandan Bandyopadhyay, Anupam Chattopadhyay, and
  Hafizur Rahaman.
\newblock A template-based technique for efficient clifford+ t-based quantum
  circuit implementation.
\newblock {\em Microelectronics Journal}, 81:58--68, 2018.

\bibitem{Maslov}
D.~{Maslov} and G.~W. {Dueck}.
\newblock Improved quantum cost for n-bit toffoli gates.
\newblock {\em Electronics Letters}, 39(25):1790--1791, 2003.

\bibitem{Abdessaied}
N.~{Abdessaied}, M.~{Amy}, M.~{Soeken}, and R.~{Drechsler}.
\newblock Technology mapping of reversible circuits to clifford+t quantum
  circuits.
\newblock In {\em 2016 IEEE 46th International Symposium on Multiple-Valued
  Logic (ISMVL)}, pages 150--155. IEEE, 2016.

\end{thebibliography}
\end{document}